\documentclass[a4paper,USenglish,cleveref, autoref, thm-restate]{lipics-v2021}

\pdfoutput=1 
\hideLIPIcs  

\nolinenumbers

\usepackage{algorithm,algpseudocode,tikz}
\usetikzlibrary{math,decorations.pathreplacing,calc}

\newcommand{\N}{{\mathbb{N}}}

\newcommand{\R}{{\mathbb{R}}}
\newcommand{\Z}{{\mathbb{Z}}}

\newcommand{\cA}{{\mathcal{A}}}

\newcommand{\cF}{{\mathcal{F}}}

\newcommand{\cP}{{\mathcal{P}}}

\newcommand{\cS}{{\mathcal{S}}}

\renewcommand{\phi}{\varphi}

\newcommand{\id}[1]{\mathbf{1}_{#1}}

\newcommand{\lra}[1]{\left\langle{#1}\right\rangle}
\newcommand{\lrp}[1]{\left({#1}\right)}

\newcommand{\lrs}[1]{\left\{{#1}\right\}}
\newcommand{\lrm}[1]{\left|{#1}\right|}

\newcommand{\lrceil}[1]{\left\lceil{#1}\right\rceil}

\newcommand{\we}{p}
\newcommand{\weh}{p^H}

\newcommand{\tx}{\tilde{x}}
\newcommand{\ty}{\tilde{y}}

\newcommand{\hx}{\hat{x}}
\newcommand{\hy}{\hat{y}}
\newcommand{\bzero}{\mathbf{0}}

\newcommand{\OP}{\operatorname{OPT}}
\newcommand{\op}{\operatorname{OPT}}
\newcommand{\opres}{\operatorname{OPT}_{\operatorname{res}}}

\title{Capacitated Partition Vertex Cover and Partition Edge Cover} 

\author{Rajni Dabas}{Northwestern University, United States}{rajni.dabas1@northwestern.edu}{https://orcid.org/0000-0001-8742-8939}{}

\author{Samir Khuller}{Northwestern University, United States}{samir.khuller@northwestern.edu}{https://orcid.org/0000-0002-5408-8023}{}

\author{Emilie Rivkin}{Northwestern University, United States \and \url{https://emmiriv.github.io/emilie-rivkin/}}{emilie.rivkin@northwestern.edu}{https://orcid.org/0009-0005-2158-6834}{}

\authorrunning{R. Dabas, S. Khuller, E. Rivkin} 

\Copyright{Rajni Dabas, Samir Khuller, and Emilie Rivkin} 

\ccsdesc[500]{Theory of computation~Packing and covering problems}

\keywords{Approximation algorithms, capacitated vertex cover, iterative rounding} 

\funding{Supported by NSF 2216970 (IDEAL) as well as a gift from Adobe Labs.}

\acknowledgements{We would like to thank Amol Deshpande, Hal Gabow, Tanmay Inamdar, Julian Mestre and Emily Pitler for useful discussions.}

\EventEditors{Mohit Singh and Tom Gur}
\EventNoEds{2}
\EventLongTitle{Approximation, Randomization, and Combinatorial Optimization. Algorithms and Techniques (APPROX/RANDOM 2026)}
\EventShortTitle{\mbox{\scriptsize APPROX/RANDOM 2026}}
\EventAcronym{APPROX/RANDOM}
\EventYear{2026}
\EventDate{August 19--21, 2026}
\EventLocation{Boston University, Boston, Massachusetts, USA}
\EventLogo{}
\SeriesVolume{392}
\ArticleNo{8}
\category{APPROX}

\begin{document}

\maketitle

\begin{abstract}
We study generalizations of the classical Vertex Cover and Edge Cover problems that incorporate group-wise coverage. 

Our first focus is the Capacitated Partition Vertex Cover (C-PVC) problem in hypergraphs. In C-PVC, we are given a hypergraph with capacities on its vertices and a partition of the hyperedge set into $\omega$ distinct groups. The objective is to select a minimum size subset of vertices that satisfies two main conditions: (1) in each group, the total number of covered hyperedges meets a specified threshold, and (2) the number of hyperedges assigned to any vertex respects its capacity constraint. A covered hyperedge is required to be assigned to a selected vertex that belongs to the hyperedge.
This formulation generalizes classical Vertex Cover, Partial Vertex Cover, and Partition Vertex Cover. 

We investigate two primary variants: soft capacitated (multiple copies of a vertex are allowed) and hard capacitated (each vertex can be chosen at most once). Let $f$ denote the rank of the hypergraph (i.e., the maximum number of vertices contained in any single hyperedge). Our main contributions are: $(i)$ an $(f+1)$-approximation algorithm for the weighted soft-capacitated C-PVC problem, which runs in polynomial time for constant \(\omega\), and $(ii)$ an $(f+\epsilon)$-approximation algorithm for the unweighted hard-capacitated C-PVC problem, which runs in $n^{O(\omega/\epsilon)}$ time. 

We also study a natural generalization of the edge cover problem, the \emph{Weighted Partition Edge Cover} (W-PEC) problem, where each edge has an associated weight, and the vertex set is partitioned into groups. For each group, the goal is to cover at least a specified number of vertices using incident edges, while minimizing the total weight of the selected edges. We present the first exact polynomial-time algorithm for the weighted case, improving runtime from \( O(\omega n^3) \) to \( O(mn+n^2 \log  n) \) and simplifying the algorithmic structure over prior unweighted approaches (that rely on the tropical matching problem). 

\end{abstract}

\newpage

\section{Introduction}

The classical \emph{Vertex Cover} (VC) problem in graphs is defined as follows: given a graph $G=(V,E)$ with weights on the vertices, find a minimum size subset of vertices $S$ so that every edge of the graph is covered, in other words, at least one of its end points is in $S$. This extremely simple problem has a rich history. It was one of the earliest problems shown to be NP-complete by Karp \cite{Karp-reducibility-72} and was shown to have a polynomial time solution when the graph is bipartite. On the hardness side, \cite{KhotRegev2008} assumes the Unique Games Conjecture and shows that Vertex Cover is hard to approximate within a factor $(2-\epsilon)$ for every $\epsilon>0$. Gavril \cite{gavril-72} showed that a simple algorithm that takes all the vertices of any maximal matching is a 2-approximation (because the cardinality of any maximal matching is a lower bound on the optimal solution size). Subsequently, this simple algorithm was extended to weighted\footnote{every vertex has a weight and the aim is to find a minimum weighted subset of vertices} graphs as well, but this required significantly deeper understanding of the role played by linear programming (LP) relaxations \cite{hochbaum-VC-82} and finally a simple combinatorial algorithm was developed based on LP duality \cite{bar-yehuda-even-81}. This problem also paved the way for improvements and the well known local ratio technique \cite{bar-yehuda-even-local-83} that found numerous applications. In addition, approximation algorithms easily follow from the work of Nemhauser and Trotter who showed that an optimal LP solution has a very simple structure \cite{nemhauser-trotter-75}. In the Fixed Parameter Tractability arena, again Vertex Cover has played a crucial role in the understanding of kernelization and efficient algorithms.

Several extensions of VC have been studied. One major direction introduces capacities on vertices, leading to the \emph{Capacitated Vertex Cover} (CVC) problem. Here, each selected vertex may cover only a limited number of incident edges. In this problem, a key distinction arises between \emph{soft capacities} (where multiple copies of a vertex can be purchased) and \emph{hard capacities} (where each vertex is chosen at most once). For soft capacities, Guha et al.\cite{guha-capVC-03} first introduced the problem and provided a 2-approximation using a primal-dual algorithm; another 2-approximation using dependent randomized rounding was later given by Gandhi et al. \cite{kortsarz-hardcapVC-06}. Notably, these results often hold for the weighted setting where vertices have costs. The hard capacity variant (VC-HC) presents significantly more challenges. Chuzhoy and Naor \cite{chuzhoy-hardcapVC-06} initiated its study, giving a 3-approximation for the unweighted version on simple graphs. 
Crucially, they also showed that the weighted version of VC-HC is at least as hard as Set Cover, which is NP-hard to approximate within a factor of $(1-\epsilon) \ln n$ for every $\epsilon>0$ \cite{DinurS14SETCOVER}. Focusing on the unweighted VC-HC, Gandhi et al. \cite{gandhi-PVC-01} improved the approximation to a tight 2. When generalizing these problems to hypergraphs, approximation guarantees typically depend on the maximum size of a hyperedge. Let $f$ denote the rank of the hypergraph, defined as the maximum number of vertices contained in any single hyperedge. A hypergraph whose rank is bounded by $f$ is commonly referred to as an $f$-hypergraph. 
For $f$-hypergraphs, Saha and Khuller \cite{SahaK12-vc} provided the first $O(f)$-approximation (specifically $\min\{6f, 65\}$). Cheung, Goemans, and Wong \cite{CheungGW14-vc} then improved the bounds to $2f$ using a deterministic two-stage LP rounding approach. The gap was finally closed to a tight $f$-approximation for unweighted VC-HC on $f$-hypergraphs, achieved independently by Kao \cite{Kao17-vc} and Wong \cite{Wong17-vc}, using iterative rounding techniques.

A separate parallel line of work relaxes the requirement to cover all edges. \emph{Partial Vertex Cover} asks to cover at least a specified number of edges~\cite{bshouty-burroughs-98,bar-yehuda-PVC-99,gandhi-PVC-01}. A further generalization is \emph{Partition Vertex Cover} (PVC), where the edge set is partitioned into $\omega$ groups, the solution must cover a specified minimum number of edges from each group. PVC was introduced by Bera et al.~\cite{partition-VC}, who obtained an \(O(\log \omega)\)-approximation via a strengthened LP with exponentially many constraints. They also proved that the problem is at least as hard as Set Cover when $\omega$ is unbounded. For $f$-hypergraphs, Hong and Kao \cite{HongKao2022} obtained a $(f+1)\log \omega$ approximation for the uncapacitated PVC problem. This is particularly relevant as the techniques in Bera et al. for graphs do not extend naturally to hypergraphs.

For constant \(\omega\), and unweighted graphs, Bandyapadhyay et al.~\cite{colorful-vertex-edge} gave a \((2+\epsilon)\)-approximation in time \(n^{O(\omega/\epsilon)}\), later improved to the weighted case with running time \(n^{O(\omega)}\) independently by Liu and Li~\cite{LiuLi-COCOON2025} and Dabas et al. \cite{DabasKR-vcec}.

Combining both dimensions - vertex capacities and partition-based coverage - yields richer and significantly more challenging models. While Partial VC with capacities has been investigated\footnote{For hard capacities, Kao et al \cite{KaoSLL19-pvc} give an \(f\)-approximation using iterative rounding. For soft capacities, Bar-Yehuda et al \cite{Bar-YehudaFMR10} give an \(f\)-approximation using the local ratio method and Mestre \cite{Mestre09-VC} gives a 2-approximation using a primal-dual algorithm.}, these methods fail when coverage requirements are partitioned across multiple groups. In this paper, we close this gap by introducing the \emph{Capacitated Partition Vertex Cover} (C-PVC) problem, which unifies the capacitated and partition variants of the Vertex Cover problem. We study the problem in the general setting of hypergraphs. 

Beyond theoretical generalization, C-PVC provides a framework for fair resource allocation in systems where demands can be served by multiple potential providers. For example, companies like Amazon solve a massive ``Capacitated Coverage'' problem to assign specific fulfillment centers to zip codes for same-day delivery. In 2016, Bloomberg published an investigation showing that Amazon’s Prime Free Same-Day Delivery excluded predominantly Black neighborhoods (like Roxbury in Boston and the Bronx in New York) while covering the white neighborhoods surrounding them \cite{Ingold2016Amazon}. 
Similar fairness concerns arise in urban mobility, where car-sharing fleets are often concentrated in high-income and advantaged zones, failing to improve accessibility in underserved communities \cite{Dill2021Equity}. 
The direct applicability of the C-PVC framework stems from the fact that these resource allocation challenges are naturally modeled as Hypergraph Vertex Cover problems rather than simple graph problems. C-PVC formulation addresses equity by enforcing a coverage threshold $\rho_g$ for each group ensuring that every demographic/geographic sector receives a minimum standard of service.

\begin{definition}[Capacitated Partition Vertex Cover (C-PVC)]
\label{def:cpvc}
We are given a hypergraph $H = (V, E)$, where $f = \max_{e \in E} |e|$ is the rank of the hypergraph. We are also given, a vertex weight function $w: V \to \R_{\ge 0}$, a vertex capacity function $k: V \to \mathbb{Z}^+$, a partition of the hyperedge set $E = E_1 \cup E_2 \cup \dots \cup E_\omega$, and a coverage requirement (threshold) $\rho_g \in \mathbb{Z}^+$ for each group $E_g$. 

We denote the capacity of a vertex $v$ by $k_v$ and its weight by $w_v$. A feasible solution consists of a set of integer copy counts $\{x_v\}_{v \in V}$ (where $x_v \in \mathbb{Z}_{\ge 0}$) and an assignment $\phi: E' \to V$ for a subset of covered edges $E' \subseteq E$, satisfying:
\begin{enumerate}
    \item \textbf{Valid Assignment:} For every edge $e \in E'$, the assigned vertex must be in the edge, i.e., $\phi(e) \in e$.
    \item \textbf{Capacity Constraint:} The number of edges assigned to any vertex $v$ is at most its total capacity:
    \[
    |\{e \in E' \mid \phi(e) = v\}| \le x_v \cdot k_v \quad \forall v \in V
    \]
    \item \textbf{Partition Constraint:} For each group $g \in [\omega]$, the number of covered edges from that group meets its threshold:
    \[
    |E' \cap E_g| \ge \rho_g
    \]
\end{enumerate}
The objective is to find a feasible solution that minimizes $\sum_{v \in V} x_v \cdot w_v$. We refer to this as the soft-capacitated variant. The hard-capacitated variant adds the constraint that $x_v \in \{0, 1\}$ for all $v \in V$.
\end{definition}

In this work, we present two algorithms for the C-PVC problem. We obtain an $(f+1)$-approximation for the soft-capacitated weighted case.
(\Cref{thm:CVC1}) and an $(f+\epsilon)$-approximation for the hard-capacitated unweighted case (\Cref{thm:CVC2}). Both algorithms run in polynomial time with respect to $n$ and $m$ but exhibit exponential dependence on \(\omega\), which yields an efficient polynomial-time solution when the number of groups $\omega$ is constant. This dependence on $\omega$ is consistent with previous work on partition-constrained problems. While the exact problem is known to be $W[1]$-hard parameterized by the number of groups $\omega$ (which rules out exact FPT algorithms), the existence of an FPT approximation algorithm parameterized by $\omega$ remains an intriguing open question. 

For the weighted variant of the hard-capacitated case, Chuzhoy and Naor \cite{chuzhoy-hardcapVC-06} showed a hardness of \(\Omega\lrp{\log n}\) of any poly-time approximation. Briefly, we can give the same upper bound for Hard C-PVC. Let \(H\) be a hypergraph of rank \(f\). For each group \(g\in[\omega]\), add one vertex \(v_g\) with capacity \(\lrm{E_g}-\rho_g\) (the allowed number of uncovered edges of that group) and \(w_{v_g}=0\). For each hyperedge \(e\in E_g\), add \(v_g\) to \(e\). Now, we can solve a full-coverage instance of weighted vertex cover with hard capacities on hypergraph \(H'\) with rank \(f+1\) using Chuzhoy and Naor's approximation algorithm for hypergraph vertex cover. This\footnote{This does not generalize to soft capacities, since we could pick unlimited copies of these additional zero weight vertices.} will give an approximation with factor \(O\lrp{\log\lrp{n+\omega}}=O(\log \lrp{n})\) for constant \(\omega\). For this reason, we only consider the unweighted variant of Hard C-PVC.

\begin{restatable}{theorem}{thmsc}
\label{thm:CVC1}
There exists a $(f+1)$-factor approximation algorithm for the \emph{Soft Capacitated Weighted Partition Vertex Cover} problem that runs in polynomial time  for constant \(\omega\).
\end{restatable}

\begin{restatable}{theorem}{thmhc}
   \label{thm:CVC2}
For any fixed \( \epsilon > 0 \), there exists an $(f+\epsilon)$-factor approximation algorithm for the \emph{Hard Capacitated Unweighted Partition Vertex Cover} problem that runs in time \ \(n^{O\lrp{\omega/\epsilon}}\). 
\end{restatable}

While the vertex cover problem focuses on selecting a subset of vertices to cover all edges, a natural dual problem is the \emph{edge cover} problem. In this variant, the goal is to select a subset of edges such that every vertex is incident to at least one selected edge. Although structurally similar to vertex cover, the edge cover problem has very different algorithmic properties—it can be solved in polynomial time via reductions to maximum matching—and it arises in applications such as network design and resource allocation.

In this paper, we study a natural generalization of the edge cover problem called the \emph{Partition Edge Cover} (Partition-EC) problem which is defined as follows.

\begin{definition}[Weighted Partition-EC]
We are given a graph \(G=\lrp{V,E}\) with a weight function \(\we:E\to\R_{\geq0}\), a partition of the vertex set \( V = V_1 \cup \dots \cup V_\omega \), and a parameter \( \rho_g \) for each group. We want to find a minimum-weight subset of edges such that at least \( \rho_g \) vertices from group \( V_g \) are covered.
\end{definition}

Bandyapadhyay et al.~\cite{bandyapadhyay_constant_2019} studied this problem in the unweighted setting and gave an exact polynomial-time algorithm with a runtime of \( O(\omega n^3) \). Their approach first reduces the problem to \emph{Budgeted Matching}, which is then reduced to \emph{Tropical Matching}, allowing the use of known algorithms for the latter to obtain an optimal solution. 

In this work, we present a simpler algorithm with improved running time (\Cref{thm:EC}). Our algorithm avoids the use of tropical matching and runs in time \({O}(mn + n^2 \log n) \). Moreover, our approach naturally extends to the more general \emph{weighted} setting, where edges have arbitrary non-negative weights.

\begin{theorem}
\label{thm:EC}
The Weighted Partition Edge Cover problem can be solved exactly in \( O(mn+n^2 \log n) \) time.
\end{theorem}

\subparagraph*{Related Work} The closely related $k$-center problem has also been well studied in the partition framework. The first algorithm was proposed by Bandyapadhyay et al.~\cite{bandyapadhyay_constant_2019}, providing a polynomial time 2-approximation while allowing \(k+\omega\) centers. A subsequent result by Anegg et al.~\cite{anegg_technique_2022} eliminated the violation in cardinality of centers and obtained a 4-approximation in time \(O(n^\omega)\), which was later improved to a 3-approximation in time \(O(n^{\omega^2})\) by Jia et al.~\cite{jia_fair_2022}. The same framework for facility location and \(k\)-median objectives has also been studied; see, for example, Dabas et al. \cite{dabas2025flofair}, Bajpai et. al. \cite{BajpaiCK25-fl}, and the references therein.

\subparagraph*{Organization of the paper} The remainder of the paper is organized as follows. We highlight the technical challenges and key ideas used to achieve our results in \Cref{sec:HLI}. Sections~\ref{sec:sc-pvc} and~\ref{sec:hc-pvc} present the ($f+1$)-approximation and $(f+\epsilon)$-approximation algorithms for soft and hard capacitated PVC, respectively.  In Section~\ref{sec:w-pec}, we describe the algorithm for W-PEC. Proofs marked by (\(\dagger\)) are deferred to the Appendix.

\subsection{Technical Challenges and Key Ideas}
\label{sec:HLI}
Covering problems are extensively studied, but partial versions of these questions (that require coverage of a given number of elements) make the problems much harder and significant new ideas and techniques were required to obtain approximation algorithms for Facility Location, $k$-median, $k$-center, Vertex Cover, etc.
Despite the seemingly ``small'' generalization from $\omega=1$ to $\omega>1$, the problems again become much harder. For example, for the $k$-center problem with outliers, there is a simple greedy algorithm with a factor 3 and a simple LP rounding algorithm with a factor 2. As mentioned in the related work, the problem for $\omega>1$ is considerably more involved; the best approximation factor of 3 requires a very high run time of $O(n^{\omega^2})$.

The same situation also arises for Vertex Cover, where for $\omega=1$ there is basically no fundamental difference between full coverage and partial coverage and a matching bound of 2 can be obtained. But for $\omega>1$ there is still no combinatorial algorithm that gives a $O(1)$ bound.

\subparagraph*{Soft C-PVC} For the soft-capacitated variant, the main challenge is managing fractional values in the LP relaxation while respecting the lower bound threshold for each group. Our algorithm first transforms an optimal fractional solution so that edge coverage is handled either by an already rounded high vertex or by a unique low vertex. We then formulate a configuration LP to handle low vertices. A configuration for a low vertex records, for each group \(g\), how many edges of \(E_g\) are served by one copy of that vertex (therefore has one block of variables for each low vertex) and \(\omega\) coverage constraints.  The key insight here is that this configuration LP has at most \(\omega\) fractional vertex blocks. We round one copy for each fractional block and retain its fractional configuration mixture - a max-flow computation determines the actual edge assignments. In the weighted case, we control the additive cost of the at most \(\omega\) rounded blocks by guessing the \(\omega\) most expensive relevant vertices. Our approach also eliminates the enumerative step for the unweighted case, thereby removing the running time's dependence on $\epsilon$ present in Bandyapadhyay et al.~\cite{colorful-vertex-edge}. While this $(f+1)$ ratio is higher than the graph-specific $(2+\epsilon)$ ratio in prior work, the removal of $1/\epsilon$ from the exponent of the runtime is a significant algorithmic simplification and efficiency gain, particularly for the weighted case.

\subparagraph*{Hard C-PVC} Our algorithm uses an iterative rounding  algorithm to fix vertices and assignments with large fractional values in each iteration. A final enumeration handles the remaining small solution space. This structure matches the known tight techniques for hard-capacitated hypergraph VC \cite{Wong17-vc} but the core technical difficulty stems from integrating the $\omega$ partition constraints into the iterative rounding framework. The algorithm in Wong \cite{Wong17-vc} crucially rely on a \emph{counting argument} applied to the final LP's extreme point solution. This argument bounds the cost of rounding the remaining ``small'' fractional vertices by bounding their total number. However, introducing $\omega$ new partition constraints fundamentally alters the structure of this extreme point, as these $\omega$ constraints can also be tight. This breaks the standard counting argument. Our key technical contribution is a modified iterative algorithm that demonstrates that the number of remaining fractional variables is bounded by a function of $\omega$. We then manage this remaining set via a final enumeration step, yielding an $(f+\epsilon)$-approximation algorithm whose runtime depends exponentially on $\omega$.

\subparagraph*{Weighted Partition-EC} Prior work by Bandyapadhyay et al.~\cite{bandyapadhyay_constant_2019} for the unweighted setting relied on a reduction to \emph{Tropical Matching}. A major challenge is that it is not clear whether the tropical matching-based approach can be extended to handle weighted instances, and the existing method requires heavy algorithmic machinery. We circumvent this by identifying a simpler structure that allows us to avoid tropical matching entirely. By reducing the problem to a simpler matching formulation, we are able to handle arbitrary weights and significantly improve the runtime to \( O(mn + n^2 \log n) \).

\section{Soft-Capacitated Weighted Partition Vertex Cover}\label{sec:sc-pvc}

We first present an integer program for the SC-PVC problem. Note that we add an upper bound of \(m\) on the multiplicity \(x_v\), since any vertex $v$ has \(k_v\geq1\) and \(\lrm{\delta(v)}\leq m\).
\begin{align*}
&\text{minimize} &\sum_{v\in V}w_vx_v & \tag{IP-SC}\label{ip:sc} \\
&\text{subject to} &\sum_{v\in e}y_{e,v}\leq 1&\qquad\forall e\in E \\
& &\sum_{e\in E_g}\lrp{\sum_{v\in e}y_{e,v}}\geq\rho_g &\qquad\forall g\in\lrs{1,\ldots,\omega} \\
& & \sum_{e\in\delta(v)} y_{e,v}\leq k_v x_v & \qquad\forall v\in V \\
& &0\leq y_{e,v}\leq x_v &\qquad\forall v\in e\in E \\
& &x_v\in\lrs{0,1,\ldots,m} &\qquad\forall v\in V \\
& &y_{e,v}\in\lrs{0,1} &\qquad \forall v\in e\in E
\end{align*}

We will say that an integral solution is \textit{feasible} if we have an integral multiplicity \(x_v\) for each vertex and an assignment function \(\phi:E\to V\cup\{\perp\}\) such that \(\lrm{\lrs{e\in E_g:\phi(e)\neq\perp}}\geq\rho_g\) for all \(g\in[\omega]\) and \(\lrm{\lrs{e\in\delta(v):\phi(e)=v}}\leq k_vx_v\).

We only need \(x\) to be integral in a feasible solution \(\lra{x,y}\) in order to find a fully integral solution. The following lemma is utilized in several prior works on capacitated vertex cover and involves a simple reduction to a network flow problem.

\begin{restatable}[\(\dagger\)]{lemma}{lemFlow}\label{lem:flow}
Let \(\lra{x,y}\) be a feasible solution to the linear relaxation of \ref{ip:sc} such that \(x\) is integral. In polynomial time, we can compute an integral solution \(\hy\) such that \(\lra{x,\hy}\) is feasible.
\end{restatable}

\subsection{The Residual LP}

We now partition the vertex set and generate a modified instance of the problem. 
\begin{itemize}
\item Let \(M\subseteq V\) with \(\lrm{M}=\omega\). For each \(v\in M\), let \(m_v\) be a fixed integer multiplicity. 
\item Define \(R= \{v\in V\setminus M:w_v\leq \min_{v\in M}w_v\ \}\).
\item Define \(Z=V\setminus\lrp{M\cup R}\).
\end{itemize}

We now give a modified linear program for each choice of $M$ and guessed multiplicty for nodes in $M$. Denote its value by \(\opres(M)\). 

Our guessing step separates the expensive part of the solution from the residual rounding problem. If the optimal solution uses at most \( \omega\) vertices, we can enumerate to find the exact solution. We may therefore assume that the optimal solution uses more than \(\omega\) distinct vertices. In this case, we guess the set \(M\) of \(\omega\) highest weight vertices and their multiplicities. Under the correct guess, every remaining optimal vertex is no more expensive than the least expensive guessed vertex, so we may safely exclude all more expensive unguessed vertices. Later, the guessed vertices also provide the budget for a set of rounded fractional vertices.


\begin{align*}
    &\text{minimize} &\sum_{v\in R}w_vx_v & \tag{\(\operatorname{LP-RES}_{\operatorfont{M}}\)}\label{lp:sc-res} \\
    &\text{subject to} &\sum_{v\in e}y_{e,v}\leq 1&\qquad\forall e\in E \\
    & &\sum_{e\in E_g}\lrp{\sum_{v\in e}y_{e,v}}\geq\rho_g &\qquad\forall g\in\lrs{1,\ldots,\omega} \\
    & & \sum_{e\in\delta(v)} y_{e,v}\leq k_vx_v & \qquad\forall v\in V \\
    & &0\leq y_{e,v}\leq x_v &\qquad\forall v\in e\in E \\
    & &0\leq x_v\leq m &\qquad\forall v\in R \\
    & & x_v=m_v &\qquad\forall v\in M \\
    & & x_v=0 &\qquad\forall v\in Z
\end{align*}

For ease of notation, we will use \(c_e:=\sum_{v\in e}y_{e,v}\) to refer to the coverage of an edge.

\begin{observation}\label{obs:res}
Let \(\op^I\) be the value of the optimal integral solution to \ref{ip:sc}. There exists a set \(M\subseteq V\) with selected multiplicites \(\lrs{m_v}_{v\in M}\) such that the resulting \ref{lp:sc-res} has integral solution of value \(\opres^I(M)\) with
\[
\opres^I(M) + \sum_{v\in M}w_vm_v = \op^I.
\]
\end{observation}

The proof of the following lemma is somewhat technical and deferred to  Subsection~\ref{subsec:ConfigLP}.

\begin{restatable}{lemma}{lemRes}\label{lem:res}
For any feasible solution \(\lra{x,y}\) to \ref{lp:sc-res}, we can construct a feasible solution \(\lra{\hx,\hy}\) such that \(\hx\) is integral and has the following properties: 
\begin{enumerate}
\item for vertices in \(M\),
\[
\sum_{v\in M}w_v\hx_v \leq f\sum_{v\in M}w_vx_v;
\]
\item for the vertices involved in \ref{lp:sc-res}, \[
\sum_{v\in R}w_v\hx_v \leq \lrp{f+1}\opres(M)+\sum_{v\in F}w_v,
\]
where \(F\subseteq R\) with \(\lrm{F}\leq\omega\).
\end{enumerate}
\end{restatable}

In order to complete the proof of Theorem \ref{thm:CVC1}, we need to show how the correct choice of \(M\subseteq V\) eliminates the term \(\sum_{v\in F}w_v\). Recall that for all \(v\in F\subseteq R\) and all \(v_j\in M\), \(w_v\leq w_{v_j}\). By combining this with Observation~\ref{obs:res}, we get
\begin{align*}
\sum_{v\in V}w_v\hx_v &= \sum_{v\in M}w_v\hx_v + \sum_{v\in R}w_v\hx_v \\
&\leq f\sum_{v\in M}w_vm_v + (f+1)\opres(M) +\sum_{v\in F}w_v \\
&\leq f\sum_{v\in M}w_vm_v + (f+1)\lrp{\op^I-\sum_{v\in M}w_vm_v} +\sum_{v\in F}w_v \\
&= (f+1)\op^I + f\sum_{v\in M}w_vm_v - (f+1)\sum_{v\in M}w_vm_v + \sum_{v\in F}w_v \\
&= (f+1)\op^I -\sum_{v\in M}w_vm_v + \sum_{v\in F}w_v \\
&\leq (f+1)\op^I.
\end{align*}
Thus our rounded solution, which is feasible with respect to both group and capacity constraints, exceeds the integral optimal solution by a factor of no more than \(f+1\). There are \(n^{O(\omega)}\) possible sets \(M\), and each set has \(m^\omega\) different sets of multiplicities. After checking all possibilities, we output the lowest-cost feasible solution; since this is the correct guess, its cost is bounded by Theorem \ref{thm:CVC1}.

Before we prove Lemma~\ref{lem:res}, we first describe a rounding scheme starting from the solution to \ref{lp:sc-res}.

\subsection{The Rounding Scheme}\label{sec:round}

We give a partition of \(V\):
\begin{itemize}
\item \(H:=M\cup\lrs{v\in R:x_v\geq\frac{1}{f+1}}\)
\item \(L:=\lrs{v\in R:0<x_v<\frac{1}{f+1}}\)
\end{itemize}

Any vertices with \(x_v=0\) are disregarded, and we set \(\hx_v=0\). Now, we round the edge variables. If any edge is such that all vertices have \(\hx_v=0\), that edge cannot be covered. For simplicity, we will remove such edges from consideration going forward (these will remain uncovered). For each \(v\in L\), create an empty set \(L_v\). Some edges may include vertices that are already rounded to \(0\). For clarity, let \(L'=L\cup\lrs{v\in V:\hx_v=0}\).

The next step preserves the total coverage \(c_e\) of every edge. It only moves this coverage between endpoints such that the coverage of each edge is partitioned between vertices in \(H\) or given to one unique vertex in $L$.

\begin{enumerate}
\item \(e\cap L'=\emptyset\): all coverage is handled by vertices in \(H\). Set \(\hy_{e,v}=y_{e,v}\) for all \(v\in e\).
\item \(e\cap H=\emptyset\): all coverage is handled by vertices in \(L\). Let \(v_\ell=\arg\max_{v\in e\cap L}y_{e,v}\). Set \(\hy_{e,v_\ell}=c_e\) and \(\hy_{e,v}=0\) \(\forall\)  \(v\in e\setminus v_\ell\). Add \(e\) to \(L_{v_\ell}\).
\item \(e\cap H,e\cap L\neq\emptyset\): coverage is distributed between vertices in \(L\) and vertices in \(H\). 
\begin{enumerate}
\item If \(y_{e,v_\ell}\geq\frac{c_e}{f+1}\), let \(v_\ell=\arg\max_{v\in e\cap L}y_{e,v}\). Set \(\hy_{e,v_\ell}=c_e\) and \(\hy_{e,v}=0\) for all other \(v\in e\). Add \(e\) to \(L_{v_\ell}\).
\item Otherwise, let \(c^\ell_e=\sum_{v\in e\cap L}y_{e,v}\) and \(c^h_e=\sum_{v\in e\cap H}y_{e,v}\). Then, \(\hy_{e,v}=y_{e,v}+\frac{y_{e,v}}{c^h_e}c^\ell_e\) for \(v\in e\cap H\) and \(\hy_{e,v}=0\) for \(v\in e\cap L\).
\end{enumerate}
\end{enumerate}

Note that, in case (3), if an endpoint in \(L\) is above the threshold \(\frac{c_e}{f+1}\), it can absorb the entire edge with an \((f+1)\)-fold increase; otherwise, the coverage provided by all endpoints in \(L\) is proportionately transfered to the endpoints in \(H\). Therefore, the threshold is carefully chosen to work for both cases. Also observe that the sets \(L_v\) are disjoint, meaning that each edge has no more than one \(v\in L\) that is responsible for its coverage. For now, we only round \(v\in H\): \(\hx_v=\lrceil{\frac{f+1}{2}x_v}\).

\begin{restatable}[\(\dagger\)]{lemma}{lemH}\label{lem:H}
For all \(v\in H\),
\[
\sum_{e\in\delta(v)}\hy_{e,v} \leq \frac{f+1}{2}k_v x_v \leq k_v\hx_v.
\]
Moreover, if \(v\in H\cap R\), \(\hx_v \leq (f+1)x_v\), and if \(v\in M\), \(\hx_v \leq f\cdot x_v.\)
\end{restatable}

We are now covering some subset of edges in each group. For \(g\in[\omega]\),
\[
h_g := \sum_{v\in H}\sum_{e\in\delta(v)\cap E_g}\hy_{e,v}.
\]
In Lemma~\ref{lem:H}, we showed that the rounded multiplicities \(\tx_v\) for \(v\in H\) can cover all rounded fractional edges that they are responsible for. Now, we have to open some vertices in \(L\) to cover the remaining \(\rho_g-h_g\) edges of each group.

\subsection{The Configuration LP}
\label{subsec:ConfigLP}

Our final task is to round the vertices in \(L\). For these vertices, we will only consider whether or not we will open one copy of each vertex. For \(v\in L\) and \(g\in[\omega]\), let \(N_{v,g}=E_g\cap L_v\) be the set of edges of group \(g\) that \(v\) is responsible for, and let \(\eta_{v,g}=\lrm{N_{v,g}}\).

Since each \(L_v\) is disjoint, some vertex \(v\) covering some edge \(e\in L_v\) has no bearing on whether another vertex \(v'\) can cover some edge \(e'\in L_{v'}\), even if \(v\) and \(v'\) are adjacent. Therefore, we consider the following question: if we were to open one copy of \(v\), how many different ways could \(v\) cover up to \(k_v\) edges in \(L_v\)? With this framing, we simplify the remaining problem by rendering edges of the same group in each \(L_v\) indistinguishable.

A feasible configuration for a vertex \(v\in L\) is a vector \(A=\lrp{A_1,\ldots,A_\omega}\in\Z^\omega_{\geq0}\) with \(A_g\leq \eta_{v,g}\) for all \(g\in[\omega]\) and \(\sum_{g=1}^\omega A_g\leq k_v\). 
The coordinate \(A_g\) records how many owned edges of group \(g\) are served by one copy. The bounds \(A_g\le\eta_{v,g}\) and \(\sum_gA_g\le k_v\) enforce edge availability and capacity, respectively. Let \(\cA_v\) be the set of all feasible configurations for \(v\in L\). We now solve this configuration LP:
\begin{align}
&\text{minimize} &\sum_{v\in L}w_v\sum_{A\in\cA_v\setminus\bzero}\lambda_{v,A} & \tag{LP-CONFIG}\label{lp:config} \\
&\text{subject to} &\sum_{A\in\cA_v}\lambda_{v,A}=1&\qquad\forall v\in L \label{eq:config}\\
& &\sum_{v\in L}\sum_{A\in\cA_v}A_g\lambda_{v,A}\geq\rho_g - h_g &\qquad\forall g\in\lrs{1,\ldots,\omega} \label{eq:grps} \\
& &\lambda_{v,A}\geq0 &\qquad\forall v\in L,A\in\cA_v \notag
\end{align}

We consider the zero configuration \(\bzero\) to mean that the vertex is not chosen, so each \(\lambda_{v,\bzero}\) does not contribute to the cost.

We first need to connect the solution of \ref{lp:config} to the solution of \ref{lp:sc-res} restricted to the vertices in \(L\). In order to do this, we will construct a feasible solution to \ref{lp:config} using the solution \(\lra{x,y}\) restricted to vertices in \(L\) and edges in \(\bigcup_{v\in L}L_v\) and the rounding scheme described in Subsection~\ref{sec:round}.

\begin{lemma}\label{lem:L}
\ref{lp:config} has a feasible solution of cost no more than \((f+1)\sum_{v\in L}w_vx_v\).
\end{lemma}
\begin{proof}
We start by considering the edges in the sets \(L_v\) for each \(v\in L\). If an edge was rounded using case (2), we have 
\[
\hy_{e,v} = c_e \leq f\cdot y_{e,v} \leq (f+1)x_v.
\]
If it was rounded using case (3a), \(v\) must be such that \(y_{e,v}\geq\frac{c_e}{f+1}\), so 
\[
\hy_{e,v} = c_e \leq (f+1)y_{e,v} \leq (f+1)x_v.
\]
We also observe that, because \(L\) contains only vertices with \(x_v<\frac{1}{f+1}\), we must have \((f+1)x_v<1\). Finally, by summing over edges in \(L_v\), we have 
\[
\sum_{e\in L_v}c_e \leq (f+1)\sum_{e\in L_v}y_{e,v} \leq (f+1)k_vx_v
\]
from the feasibility of the LP solution. This implies that 
\[
\sum_{e\in L_v}\frac{c_e}{(f+1)x_v} \leq k_v.
\]

Consider the following set of vectors:
\[
\cP(d,k)= \lrs{z\in\R^d:0\leq\|z\|_\infty\leq1, \|z\|_1\leq k}.
\]
This defines a concave set with integral extremal points. Assume that we ordered the edges of \(L_v\). Then, the vector
\[
z_v = \lrp{\frac{c_{e_1}}{(f+1)x_v},\frac{c_{e_2}}{(f+1)x_v},\ldots}
\]
is contained within \(\cP(|L_v|,k_v)\). Let \(\cS_v=\lrs{S\subseteq L_v:|S|\leq k_v}\). We can now use the set of indicators \(\id{S}\) of \(S\in\cS_v\) as a basis for \(z_v\). Write \(z_v=\sum_{S\in\cS_v}\zeta_{v,S}\id{S}\) for \(\sum_{S\in\cS_v}\zeta_{v,S}=1\).

Next, we want to move from the vectors \(\id{S}\) to feasible edge configurations. Define the projection \(\pi_v:\R^{|L_v|}\to\R^\omega\) such that column \(g\) is \(\id{N_{v,g}}\). Each \(\pi_v\id{S}\) is a feasible edge configuration. For each \(A\in\cA_v\setminus\bzero\), let \(S_A=\lrs{S\in\cS_v:\pi_vS=A}\). Define \(\lambda_{v,A}=(f+1)x_v\sum_{S\in S_A}\zeta_{v,S}\) and \(\lambda_{v,\bzero}=1-\sum_{A\neq\bzero}\lambda_{v,A}\). Therefore, \(\lambda\) satisfies the constraints \eqref{eq:config}.

Now, we look at the constraints \eqref{eq:grps}. For one edge, \(\sum_{S\ni e}\zeta_{v,S}=\frac{c_e}{(f+1)x_v}\). We know that the edges are sufficiently covered in the rounded solution, so
\begin{align*}
\sum_{v\in L}\sum_{A\in\cA_v}A_g\lambda_{v,A} &= \sum_{v\in L}\sum_{A\in\cA_v} A_g\sum_{S\in S_A}\zeta_{v,S}(f+1)x_v \\
&= \sum_{v\in L}(f+1)x_v\sum_{S\in\cS} \id{N_{v,g}}\cdot\id{S}\zeta_{v,S} \\
&= \sum_{v\in L}\sum_{S\in\cS_v}\sum_{e\in S\cap N_{v,g}}(f+1)x_v\zeta_{v,S} \\
&= \sum_{v\in L}(f+1)x_v\sum_{e\in N_{v,g}}\sum_{S\ni e}\zeta_{v,S} \\
&= \sum_{v\in L}(f+1)x_v\sum_{e\in N_{v,g}}\frac{c_e}{(f+1)x_v} \\
&= \sum_{v\in L}\sum_{e\in N_{v,g}}c_e \geq \rho_g-h_g
\end{align*}

Now that the projection defines a feasible solution to \ref{lp:config}, we need to show that the objective function value is bounded:
\begin{align*}
\sum_{v\in L}w_v\sum_{A\in\cA_v\setminus\bzero}\lambda_{v,A} &= \sum_{v\in L}w_v\sum_{A\in\cA_v\setminus\bzero}(f+1)x_v\sum_{S\in S_A}\zeta_{v,S} \\
&= \sum_{v\in L}w_v(f+1)x_v\sum_{S\in\cS\setminus\bzero}\zeta_{v,S} \\
&\leq (f+1)\sum_{v\in L}w_vx_v. \qedhere
\end{align*}

\end{proof}

Let \(F\subseteq L\) be the set of vertices with fractional \(\lambda_{v,A}\). While we also want to open this set of vertices, bounding the cost is made difficult by the lack of an integral configuration.

\begin{lemma}\label{lem:F}
Any optimal basic feasible solution \(\lambda\) to \ref{lp:config} has \(\lrm{F}\leq\omega\).
\end{lemma}
\begin{proof}
We start by turning constraints \eqref{eq:grps} into equalities. Our new LP is 
\begin{align*}
&\text{minimize} &\sum_{v\in L}w_v\sum_{A\in\cA_v\setminus\{0\}}\lambda_{v,A} & \\
&\text{subject to} &\sum_{A\in\cA_v}\lambda_{v,A}=1&\qquad\forall v\in L \\
& &\sum_{v\in L}\sum_{A\in\cA}A_g\lambda_{v,A}=\rho_g-h_g+s_g &\qquad\forall g\in\lrs{1,\ldots,\omega} \\
& &\lambda_{v,A}\geq0 &\qquad\forall v\in L,A\in\cA_v \\
& &s_g\geq0 &\qquad\forall g\in\lrs{1,\ldots,\omega}
\end{align*}
A solution to this LP has up to \(\lrm{L}+\omega\) non-tight variables, meaning that up to \(\lrm{L}+\omega\) of the variables \(\lambda_{v,A}\) and \(s_g\) are positive. Constraints \eqref{eq:config} ensures that at least one \(\lambda_{v,A}\) per \(v\in L\) is positive, and at least \(\lrm{L}-\omega\) of those constraints has only one positive \(\lambda_{v,A}\), so up to \(\omega\) of these constraints can have fractional \(\lambda_{v,A}\).
\end{proof}

We now have the machinery to prove Lemma~\ref{lem:res}. Recall the statement:

\lemRes*

\begin{proof}
For each \(v\in L\), if \(\lambda_{v,\bzero}=1\), set \(\hx_v=0\) and \(\hy_{e,v}=0\) for all \(e\in L_v\). Otherwise, set \(\hx_v=1\). If \(v\notin F\), there is one integral edge configuration. Select a set \(N(A)\) consisting of \(A_g\) edges from each \(N_{v,g}\), and set \(\hy_{e,v}=\id{e\in N(A)}\). If \(v\in F\), let \(\cF_v=\lrs{A\in\cA_v:\lambda_{v,A}}>0\). For each configuration \(A\in\cF_v\), define a set of edges \(N(A)\). Then, set \(\hy_{e,v}=\sum_{A\in\cF_v}\lambda_{v,A}\id{e\in N(A)}\). 

Since the solution \(\lambda\) is feasible for \ref{lp:config}, 
\begin{align*}
\sum_{v\in L}\sum_{e\in N_{v,g}}\hy_{e,v} &= \sum_{v\in L\setminus F} \sum_{e\in N_{v,g}}\id{e\in N(A)} + \sum_{v\in F}\sum_{A\in\cF_v}\sum_{e\in N_{v,g}}\lambda_{v,A}\id{e\in N(A)} \\
&= \sum_{v\in L}A_g+\sum_{v\in F}\sum_{A\in\cF_v}\sum_{e\in N_{v,g}}\lambda_{v,A}A_g \geq \rho_g-h_g,
\end{align*}
since the \(L_v\) are disjoint and no edge can be double-counted. Moreover, if \(v\in L\setminus F\),
\[
\sum_{e\in\delta(v)}\hy_{e,v} = \sum_{g=1}^\omega A_g \leq k_v=k_v\hx_v.
\]
If \(v\in F\),
\[
\sum_{e\in\delta(v)}\hy_{e,v} = \sum_{e\in\delta(v)}\sum_{A\in\cF_v}\lambda_{v,A}\id{e\in N(A)} = \sum_{e\in\delta(v)}\sum_{A\in\cF_v}\lambda_{v,A}A_g \leq k_v=k_v\hx_v.
\]
From Lemma~\ref{lem:H}, the vertices in \(H\) respect the capacity constraints. Thus, we can conclude that \(\lra{\tx,\ty}\) over all \(v\in L\cup H\) satisfies both the group and capacity constraints. We now bound the cost of the rounded vertices. We use Lemma~\ref{lem:H}. For part (1), we can sum over the bound on \(\hx_v\) for all vertices in \(M\). For part (2),
\[
\sum_{v\in R\cap H}w_v\hx_v \leq (f+1)\sum_{v\in R\cap H}w_vx_v.
\]
From Lemma~\ref{lem:L},
\[
\sum_{v\in L}w_v\hx_v \leq (f+1)\sum_{v\in L}w_vx_v+\sum_{v\in F}w_v.
\]
By combining these statements, we have 
\[
\sum_{v\in R}w_v\hx_v \leq (f+1)\sum_{v\in R}w_vx_v + \sum_{v\in F}w_v = (f+1) \opres(M) + \sum_{v\in F}w_v,
\]
which concludes the proof of Lemma~\ref{lem:res}.

\end{proof}

\subsection{The Full Algorithm}

We end this section with an overview of the algorithm and a discussion of the runtime. 

\begin{enumerate}
\item Enumerate through every set of vertices of size at most \(\omega\) and every integral multiplicity from \(1\) to \(m\). If any set is feasible, take the cheapest feasible set. Otherwise, continue to the main algorithm. There are \(\lrp{|V|+|E|}^{O(\omega)}\) such sets.
\item For every set \(M\subseteq V\) with \(|M|\leq\omega\) and every integral multiplicity from \(1\) to \(m\), solve \ref{lp:sc-res}. There are \(\lrp{|V|+|E|}^{O(\omega)}\) such sets.
\begin{enumerate}
    \item Round the solution using Lemma~\ref{lem:H}.
    \item Solve \ref{lp:config} for the vertices in \(L\). The LP has up to \(|V|(|E|+1)^\omega\) variables and \(|V|+\omega\) constraints.
    \item Round the fractional vertices using Lemmas~\ref{lem:L} and~\ref{lem:F}.
    \item Continue with the lowest-cost feasible solution, which corresponds to the correct guess of \(M\).
\end{enumerate}
\item Find an integral edge assignment using Lemma~\ref{lem:flow}. The network has \(O(|V|+|E|+\omega)\) nodes and \(O(|V|+|E|+\omega)\) arcs.
\end{enumerate}

Each step requires time that is polynomial in the \(|V|+|E|\), and we repeat each step a number of times that is polynomial in \(|V|+|E|\). Since \(\omega\) is assumed to be constant, the overall algorithm is polynomial in \(|V|+|E|\). Note that this runtime is independent of the hyperedge size \(f\). This completes the proof of Theorem~\ref{thm:CVC1}.

\section{Hard Capacitated Unweighted Partition Vertex Cover}
\label{sec:hc-pvc}

In this section, we study the Hard Capacitated Unweighted Partition Vertex Cover problem in hypergraphs using iterative rounding. Before detailing the iterative steps, we define the LP relaxation using fractional vertex variables $x$ and assignment variables $y$. These variables are constrained to satisfy the partition thresholds and vertex capacities while ensuring that no edge assignment exceeds the fractional status of its endpoints. This formulation minimizes the total fractional vertex count, establishing a baseline solution that informs the subsequent rounding iterations. The integer program closely resembles \ref{ip:sc}, so we will not restate it. The formal relaxation is presented below:  

\begin{align*}
\text{minimize} \quad & \sum_{v \in V} x_v \\
\text{subject to:} \quad & \\
& \sum_{v \in e} y_{e,v} \le 1 & & \forall e \in E \\
& \sum_{e \in E_g} \left( \sum_{v \in e} y_{e,v} \right) \ge \rho_g & & \forall g\in\lrs{1,\ldots,\omega} \\
& y_{e,v} \le x_v & & \forall v\in e\in E \\
& \sum_{e \in \delta(v)} y_{e,v} \le k_v x_v & & \forall v \in V \\
& 0 \le x_v \le 1 & & \forall v \in V \\
& y_{e,v} \ge 0 & & \forall v\in e\in E
\end{align*}

It is sufficient to find an integral solution for $x$; an integral assignment $y$ can then be found efficiently via a max-flow computation. Our algorithm therefore focuses on finding an integral $x$ with a bounded cost.

\subsection{Algorithm}
Our algorithm is based on an iterative rounding framework. The core idea is to repeatedly solve an LP relaxation, fix a subset of variables to integral values based on the solution, and then re-solve a smaller, updated LP on the remaining problem. This process is guided by a dynamic LP formulation that adapts to the decisions made in each step (see \ref{lp:ip}). 

We maintain several sets for the iterative rounding algorithm. Let $D$ be the set of vertices whose final integral counts have been determined, and let $T_v^g$ denote the set of edges of group $g$ that have been definitively assigned to vertex $v$. We define $T = \bigcup_{v \in V} \bigcup_{g=1}^\omega T_v^g$ as the set of all covered edges. 
Based on the fractional solution from the previous step, we partition the remaining undecided vertices into two groups: $U = \{v \in V \setminus D \mid x_v^* \ge 1/f\}$, which contains vertices with high fractional values, and $W = \{v \in V \setminus D \mid 0 < x_v^* < 1/f\}$, which contains those with low fractional values. $U$ is further divided into $U_> = \{v \in V \setminus D \mid x_v^* > 1/f\}$ and $U_= = \{v \in V \setminus D \mid x_v^* = 1/f\}$. Let $Z = \{v \in V \setminus D \mid x_v^* = 0 \}$. Initially the sets $T$ and $D$ are empty.

\begin{align*}
\text{minimize} \quad & \sum_{v \in V \setminus D} x_v \tag{LP-IP}\label{lp:ip}\\
\text{subject to:} \\
& \sum_{v \in e \setminus D} y_{e,v} \le 1 - \sum_{v' \in e \cap D} \bar{y}_{e,v'} & & \forall e \in E \setminus T \\
& \sum_{e \in E_g \setminus T} \sum_{v \in e \setminus D} y_{e,v} \ge \rho_g & & \\
& \qquad\qquad - \sum_{v \in V} |T_v^g| - \sum_{e \in E_g \setminus T} \left(\sum_{v' \in e \cap D} \bar{y}_{e,v'}\right) & & g\in\lrs{1,\ldots,\omega}\\
& y_{e,v} \le x_v & & \forall e \in E \setminus T, v \in e \setminus D \\
& \sum_{e \in \delta(v) \setminus T} y_{e,v} \le \left(k_v - \sum_{g \in \omega} |T_v^{g}| \right) x_v & & \forall v \in U_> \\
& \sum_{e \in \delta(v) \setminus T} y_{e,v} \le k_v x_v & & \forall v \in W \\
& 1/f \le x_v \le 1 & & \forall v \in U_> \\
& 0 \le x_v \le 1/f & & \forall v \in W \\
& y_{e,v} \ge 0 & & \forall e \in E \setminus T, v \in e \setminus D
\end{align*}

Note that there is no capacity constraint included for vertices in $U_=$ within \ref{lp:ip}. This is because any vertex $u$ entering $U_=$ is immediately processed by the algorithm (its value is fixed to $\bar{x}_u = 1$ and it is moved to the decided set $D$). Thus, it is no longer an active variable in the updated LP over $V \setminus D$, and omitting it does not cause the feasible region to grow. The complete iterative rounding framework is detailed in Algorithm~\ref{alg:my-cr-vchc}. It formalizes the process of identifying which variables to fix, updating the problem state, and iterating until a fully integral solution is found.

\begin{algorithm}
\footnotesize
\caption{Iterative Rounding Algorithm}
\label{alg:my-cr-vchc}
\begin{algorithmic}[1]
    \State Solve the initial LP Relaxation to get an extreme point solution $\lra{x^*, y^*}$.
    \State $T^g_v \gets \emptyset$ for all $v \in V$ (which implicitly sets $T = \emptyset$), $g\in\lrs{1,\ldots,\omega}$  and $D \gets \emptyset$.
    \State Initialize $U, U_>, U_=, W, Z$ based on $\lra{x^*, y^*}$.
    
    \Repeat
        \For{$v \in Z$} \Comment{Handle zero-valued vertices}
            \State Set $\bar{x}_v \gets 0$.
            \State $D \gets D \cup \{v\}$.
        \EndFor
        
        \If{$y^*_{e,u} = x_u^*$ for some $u \in U$ and $e \in \delta(u) \setminus T$} \Comment{Assign tight edges}
            \State Set $\bar{y}_{e,u} \gets 1$, and $\bar{y}_{e,v} \gets 0$ for $v \in e \setminus \{u\}$.
            \State $T^g_u \gets T^g_u \cup \{e\}$ \textbf{for} $ e\in E_g$.
        \EndIf
        
        \For{$u \in U_=$} \Comment{Handle boundary-valued vertices}
            \State Set $\bar{x}_u \gets 1$.
            \State \textbf{Set $\bar{y}_{e,u} \gets y^*_{e,u}$ for $e \in \delta(u) \setminus T$}. \Comment{Fix edge assignments}
            \State $D \gets D \cup \{u\}$.
        \EndFor
        \State Solve the updated LP-IR for a new solution $\lra{x^*, y^*}$.
        \State Update $U, U_>, U_=, W, Z$ based on the new $\lra{x^*, y^*}$.
        
    \Until{$U_= = \emptyset$ and $Z = \emptyset$ and no tight edges are found}
    
    \For{$v \notin D$} \Comment{Final rounding step}
        \State Set $\bar{x}_v \gets \lceil x_v^* \rceil$.
    \EndFor
    
    \State \textbf{return} integral solution $\bar{x}$.
\end{algorithmic}
\end{algorithm}

Note that when a tight edge is found and permanently assigned (Lines 8-11), we intentionally do not fix $x_u$ to 1. The variable $x_u$ remains active and fractional in the updated LP so its residual capacity can be utilized for other incident edges, and to prevent a premature increase in the LP objective. Since $u \in U$, Lemma 6 guarantees its fractional value remains at least $1/f$, ensuring it will be safely rounded to 1 at the end of the algorithm.

\subsection{Analysis}
\begin{lemma}
\label{lem:intx}
Suppose a vertex $u$ satisfies $x_u^* \ge 1/f$ at some iteration. Then $x_u^* \ge 1/f$ in all subsequent iterations as long as $u \notin D$. Consequently, the final integral value $\bar{x}_u$ will be at least 1.
\end{lemma} 

\begin{proof}
When $x_u^* \ge 1/f$, the vertex $u$ is placed in the set $U$. We consider two cases based on its value.
\begin{itemize}
    \item If $x_u^* = 1/f$, then $u \in U_=$. In the same iteration, the algorithm adds $u$ to the set of decided vertices $D$ and fixes its final integral value as $\bar{x}_u = 1$.
    \item If $x_u^* > 1/f$, then $u \in U_>$. For all subsequent iterations where $u$ remains an undecided variable, the updated LP-IR includes the explicit constraint $x_u \ge 1/f$. This ensures its fractional value cannot drop below $1/f$.
\end{itemize}
Eventually, either $u$ enters $U_=$ (and is rounded to $\bar{x}_u = 1$) or the algorithm terminates with $u \notin D$. In the latter case, it is rounded up in the final step, yielding $\bar{x}_u = \lceil x_u^* \rceil \ge \lceil 1/f \rceil = 1$. In all possible outcomes, if $u$ ever enters $U$, its final value $\bar{x}_u$ is at least 1.
\end{proof}

\begin{lemma}
\label{lem:fes}
The solution $\bar{x}$ returned by the algorithm is a feasible integral solution to the problem.
\end{lemma}

\begin{proof}
The final solution $\bar{x}$ is integral by construction. Variables are set to integer values (0 or 1) during the loop, and the final step applies the ceiling function.

Next, the initial LP ensures that the group coverage $\rho_g$ is satisfied fractionally. As the algorithm fixes variables and assigns edges to $T$, it reduces the coverage requirement for subsequent LPs. The final LP solution $\lra{x^*, y^*}$ satisfies the remaining coverage fractionally. The final rounding step, $\bar{x}_v = \lceil x_v^* \rceil$, only increases the available capacity, which means at least as many edges can be covered as in the fractional solution. Thus, the integral solution meets or exceeds the fractional coverage and satisfies all group coverage requirements.

Finally, we consider capacity constraints for each group of vertices:
\begin{itemize}
\item \textbf{For vertices in $U_>$:}  The capacity used by any ``tight'' edges is explicitly subtracted and budgeted for in later steps. When the remaining vertices are rounded up, the new integer capacity is always greater than or equal to the fractional capacity required by the LP solution ($k_u \lceil x_u^* \rceil \ge k_u x_u^*$), ensuring the constraint holds.
\item \textbf{For vertices in $U_=$:} When a vertex is rounded from $1/f$ to $1$, the LP guarantees that its fractional edge assignments required a capacity of at most $(k_u - \sum_{g \in [\omega]}|T^g_u|)/f$. This is well within the newly allocated integer capacity of $(k_u - \sum_{g \in [\omega]}|T^g_u|)$, so the assignment is valid.
\item \textbf{For vertices in $W$ :} At the end of the algorithm, these vertices are rounded up from a value less than $1/f$ to $1$. The LP solution's required capacity was less than the final integer capacity ($k_w x_w^* < k_w \cdot 1$), so the final assignment is feasible. \qedhere
\end{itemize}
\end{proof}

\begin{lemma}
\label{lem:count}
    Let $\lra{x^*, y^*}$ be the extreme point solution upon which the algorithm terminates. $|W|\leq |U_1|+\omega$ where $U_1=\{u \in U: x^*_u=1 \}$.
\end{lemma}

\begin{proof}
The proof is by a counting argument on the variables and constraints of the final LP-IR. As an extreme point solution, it is uniquely defined by a system of linear equations formed by setting a subset of the LP's inequality constraints to be tight. A fundamental property of such solutions is that the number of variables must equal the number of linearly independent tight constraints.

At termination, the set of undecided vertices is $V \setminus D = U_> \cup W$. The algorithm's termination conditions ensure that for any $u \in U_>$, we have $x_u^* > 1/f$ and $y^*_{e,u} < x_u^*$. For any $w \in W$, we have $0 < x_w^* < 1/f$.

The variables in the final LP are the $\{x_v\}$ for undecided vertices and the associated $\{y_{e,v}\}$ assignment variables. The total number of variables is 
\[
|U_>| + |W| + \sum_{e \in E \setminus T} |e \setminus D| 
    = |U_>| + |W| + \sum_{e \in E \setminus T} |e \cap U_>| + \sum_{e \in E \setminus T} |e \cap W|.
\]

Next, we establish an upper bound on the number of linearly independent constraints that can be tight. For the \textbf{partition constraints}, \textbf{capacity constraints for $U_>$:}, and \textbf{upper bounds $x_u \le 1$:} at most $\omega$, \(|U_>|\), and \(|U_1|\) constraints can be tight.

The \textbf{constraints involving $W$ vertices} include the capacity constraint for each $w \in W$ ($\sum y_{e,w} \le k_w x_w$), the coupling constraints ($ y_{e,w} \le x_w$), and the non-negativity constraints ($y_{e,w} \ge 0$). The maximum number of constraints from this group is limited for two reasons. First, for any specific edge $e$, the constraints $y_{e,w} \ge 0$ and $y_{e,w} \le x_w$ cannot both be tight simultaneously. This is because $ x_w^* > 0$. Thus, for each edge, at most one of these two boundary constraints can be part of the basis. Second, the entire collection of constraints for vertex $w$ is linearly dependent. For a given vertex $w$, if the constraints corresponding to $y_{e,w}$ are chosen to be all tight, the tightness of the capacity constraint is a direct consequence and is therefore redundant; it cannot add a new linearly independent equation to the basis. Due to these two properties, the maximum number of independent tight constraints associated with any single vertex $w$ is $|\delta(w) \setminus T|$. Summing over all vertices in $W$ establishes the total contribution of at most 
\[
\sum_{w \in W} |\delta(w) \setminus T| = \sum_{e \notin T} |e \cap W|
\]
constraints to the basis.

Finally, in order to handle the \textbf{constraints involving $U_>$ vertices}, we do a per-edge analysis. For each edge $e \in E \setminus T$, we analyze the relationship between its edge coverage constraint and its non-negativity constraints for $U_>$ endpoints. Let $C_e$ be the edge coverage constraint
\[
\sum_{v \in e \setminus D} y_{e,v} \le 1 - \sum_{v' \in e \cap D} \bar{y}_{e,v'},
\]
and let $\{N_{e,u}\}$ be the set of non-negativity constraints: $y_{e,u} \ge 0$ for all $u \in e \cap U_>$.

\begin{claim*}
For any edge $e$ with at least one endpoint in $U_>$, it is impossible for $C_e$ and all constraints in $\{N_{e,u}\}$ to be tight simultaneously.
\end{claim*}

\begin{claimproof}
Assume for contradiction that for an edge $e$ where $|e \cap U_>| \ge 1$, constraint $C_e$ and all constraints $N_{e,u}$ are tight.
The tightness of all $N_{e,u}$ implies $y_{e,u} = 0$ for all $u \in e \cap U_>$. Thus, the edge's total coverage comes only from $W$ vertices: 
\[
\sum_{v \in e \setminus D} y_{e,v} = \sum_{w \in e \cap W} y_{e,w}.
\]
Since $y_{e,w} < 1/f$ for any $w \in W$, we have a strict upper bound: 
\[
\sum_{v \in e \setminus D} y_{e,v} < |e \cap W|/f.
\]
The tightness of $C_e$ means 
\[
\sum_{v \in e \setminus D} y_{e,v} = 1 - \sum_{v' \in e \cap D} \bar{y}_{e,v'}.
\]
This right-hand side has a lower bound of $|e \setminus D|/f$. Combining these gives 
\[
|e \setminus D|/f \le \sum_{v \in e \setminus D} y_{e,v} < |e \cap W|/f.
\]
This implies \[
|e \cap U_>| + |e \cap W| < |e \cap W|,
\]
which simplifies to $|e \cap U_>| < 0$, a contradiction.
\end{claimproof}

This claim implies that for each edge $e$, the set of $1 + |e \cap U_>|$ constraints (consisting of $C_e$ and $\{N_{e,u}\}$) are linearly dependent when all are tight. Therefore, at most $|e \cap U_>|$ of them can be chosen for a linearly independent basis. The total number of tight constraints from this group, summed over all edges, is at most $\sum_{e \in E \setminus T} |e \cap U_>|$.

By equating the number of variables to the upper bound on the number of tight, linearly independent constraints, we have: 
\[
|U_>| + |W| + \sum_{e \in E \setminus T} |e \cap U_>| + \sum_{e \in E \setminus T} |e \cap W| \le \omega + |U_>| + |U_1| + \sum_{e \in E \setminus T} |e \cap W| + \sum_{e \in E \setminus T} |e \cap U_>|. 
\]
Canceling the common terms from both sides yields the desired result.
\end{proof}

We can now return to the proof of \Cref{thm:CVC2}. 

\thmhc*

\begin{proof}

The feasibility of the integral solution $\bar{x}$ produced by the algorithm is guaranteed by \Cref{lem:fes}. We now bound its cost.

Let $\OP_{LP}$ be the optimal value of the initial LP relaxation. First, observe that the optimal value of the LP relaxation does not decrease between iterations. When rounding decisions are made in the main loop, the problem is modified by either fixing variables or enforcing stricter inequalities (e.g., $x_v \geq 1/f$ instead of $x_v \geq 0$). Because we shrink the feasible region of a minimization problem, the objective value cannot improve. The cost added to our integral solution during the main loop can be bounded by charging it to the fractional value removed from the LP. Specifically, when $x_u^* \geq 1/f$ is rounded to $1$, the cost ratio is at most $1/(1/f) = f$.

Let $S_L$ be the set of vertices finalized during the iterative loop and $S_F = U_> \cup W$ be the set of vertices rounded at the end. An accounting argument based on the above observation shows that the cost from vertices rounded during the loop, $|S_L|$, plus $f$ times the cost of the final LP, $f \cdot \OP_F$, is at most $f \cdot \OP_{LP}$.

The cost from the final rounding step is $\text{Cost}(S_F) = \sum_{v \in U_> \cup W} \lceil x_v^* \rceil = |U_>| + |W|$. By \Cref{lem:count}, this is bounded by $|U_>| + |U_1| + \omega$. For any $u \in U_>$, we have $1 < f \cdot x_u^*$. This is sufficient to show that, for \(f\geq2\),
\[
|U_>| + |U_1| \le f \cdot \sum_{u \in U_>} x_u^* \le f \cdot \OP_F.
\]
Therefore, the cost of the final step is bounded as
\[
\text{Cost}(S_F) = |U_>| + |W| \le f \cdot \OP_F + \omega
\]
Combining the costs from both phases gives the total cost
\[
\text{Cost}(\bar{x}) = |S_L| + \text{Cost}(S_F) \le |S_L| + f \cdot \OP_F + \omega
\]
Since $|S_L| + f \cdot \OP_F \le f \cdot \OP_{LP}$, we have a solution with a total cost bounded by $f \cdot \OP_{LP} + \omega$.

We enumerate all possible solutions of size $\kappa = 1, 2, \dots, \lceil \omega/\epsilon \rceil$. This step takes $n^{O(\omega/\epsilon)}$ time. If a feasible solution is found, the first one encountered must be an optimal solution, since we are checking solution sizes in increasing order, and we are done.

Otherwise, if this process is completed without finding a feasible solution, we know that the optimal integral solution, $\OP$, has a cost greater than $\omega/\epsilon$. In this case, we run our iterative rounding algorithm. By the analysis above, it returns a feasible solution with a cost bounded by
\[
\text{Cost}(\bar{x}) \le f \cdot \OP_{LP} + \omega
\]
Since $\OP_{LP} \le \OP$ and we have established that $\OP > \omega/\epsilon$, we can bound the cost of our solution as
\[
\text{Cost}(\bar{x}) < f \cdot \OP + \epsilon \cdot \OP = (f + \epsilon) \cdot \OP.
\]
This provides an $(f+\epsilon)$-approximation for the problem.
\end{proof}

\section{Weighted Partition Edge Cover}
\label{sec:w-pec}
We consider the following generalization from \cite{colorful-vertex-edge}:

\begin{definition}[Weighted Budgeted Matching]
We are given a graph \(G=\lrp{V,E}\) with a weight function \(\we:E\to\R_{\geq0}\), a partition of the vertex set \( V = V_1 \cup \dots \cup V_\omega \), and a parameter \( \rho_g \) for each group. We want to find a minimum-weight matching such that at least \( \rho_g \) vertices from group \( V_g \) are matched.
\end{definition}

First, we note that the reduction between Partition-EC and  Budgeted Matching in Lemma 5 \cite{colorful-vertex-edge} holds in the weighted case. 

\begin{lemma}[\cite{colorful-vertex-edge}] \label{lem:PEC-WBM}
If Weighted Budgeted Matching can be solved in time \(T\lrp{n,m}\), then Partition-EC can be solved in time \(T\lrp{2n,m+n}+O\lrp{m+n}\). Here $n$ and $m$ refer to the number of nodes and edges respectively.
\end{lemma}

In order to be matched or covered at minimum cost, a vertex will pay for its least expensive adjacent edge, so we can modify the reduction by giving the edge between each vertex and its auxillary vertex this minimum adjacent weight. Moreover, we eliminate the dependence on \(\omega\).

Now, we give a simple reduction to the standard problem of finding a maximum weight matching in a graph. This allows us to generalize the problem to weighted graphs, improve the running time, and avoid the reduction to tropical matching completely. Our algorithm also admits the use of standard libraries for weighted matching.

\begin{enumerate}
\item Let \(M\) be large with \(M> \sum_e \we_e \). 
\item Let $H$ be a copy of $G$ with the following vertex sets added. For each group $g$ add a set $V'_g$ of $|V_g|-\rho_g$ nodes. Add edges between all nodes in $V'_g$ and $V_g$ for each $g$. The weight of an edge $e$ in $H$ is $2M-\we_e$ when $e \in E$, and the weight of an edge $(x,y)$ where $x \in V'_g, y \in V_g$ is $M$. See Figure \ref{fig:reduction} for an example
\item Find a maximum weight matching \( \mathcal{M}_H \) in $H$ using the algorithm by Gabow \cite{gabow-90,Gabow18} 
The matched edges induced by the nodes of $G$ should be a minimum weight matching satisfying the group requirements. If the weight of the maximum weight matching is at most $(n-1)M$ then no solution exists.
\end{enumerate}

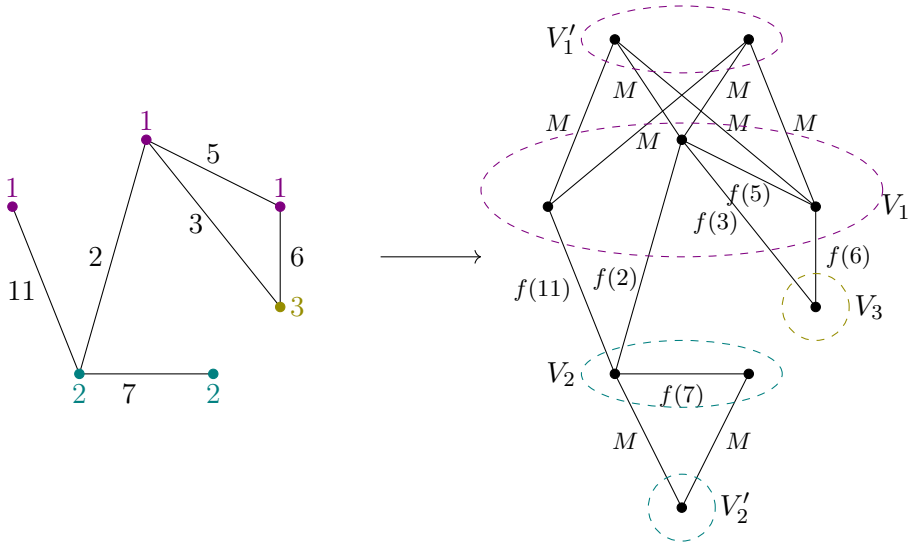
\begin{figure}[ht]
    \centering
    \begin{tikzpicture}[scale=0.89]

\coordinate (b1) at (-2,0);
\coordinate (b2) at (0,1);
\coordinate (b3) at (2,0);
\coordinate (g1) at (-1,-2.5);
\coordinate (g2) at (1,-2.5);
\coordinate (r1) at (2,-1.5);

\draw (b1) -- (g1) node[midway, left] {\footnotesize 11};
\draw (g1) -- (b2) node[midway, left] {\footnotesize 2};
\draw (b2) -- (b3) node[midway, above] {\footnotesize 5};
\draw (g1) -- (g2) node[midway, below] {\footnotesize 7};
\draw (b3) -- (r1) node[midway, right] {\footnotesize 6};
\draw (b2) -- (r1) node[midway, left] {\footnotesize 3};

\filldraw[color=violet] (b1) circle (2pt) node[above] {1};
\filldraw[color=violet] (b2) circle (2pt) node[above] {1};
\filldraw[color=violet] (b3) circle (2pt) node[above] {1};
\filldraw[color=teal] (g1) circle (2pt) node[below] {2};
\filldraw[color=teal] (g2) circle (2pt) node[below] {2};
\filldraw[color=olive] (r1) circle (2pt) node[right] {3};

\draw[->] ($(b3)!0.5!(r1)+(1.5,0)$) -- ($(b3)!0.5!(r1)+(3,0)$);

\coordinate (adj) at (8,0);
\coordinate (mb1) at ($(b1)+(adj)$);
\coordinate (mb2) at ($(b2)+(adj)$);
\coordinate (mb3) at ($(b3)+(adj)$);
\coordinate (mg1) at ($(g1)+(adj)$);
\coordinate (mg2) at ($(g2)+(adj)$);
\coordinate (mr1) at ($(r1)+(adj)$);

\coordinate (sb1) at ($(mb1)!0.5!(mb2)+(0,2)$);
\coordinate (sb2) at ($(mb3)!0.5!(mb2)+(0,2)$);
\coordinate (sg1) at ($(mg1)!0.5!(mg2)-(0,2)$);

\draw (mb1) -- (mg1) node[midway, left] {\footnotesize $89$};
\draw (mg1) -- (mb2) node[midway, below left] {\footnotesize $98$};
\draw (mb2) -- (mb3) node[midway, below] {\footnotesize $95$};
\draw (mg1) -- (mg2) node[midway, below] {\footnotesize $93$};
\draw (mb3) -- (mr1) node[midway, right] {\footnotesize $94$};
\draw (mb2) -- (mr1) node[midway, left] {\footnotesize $97$};

\draw (mb1) -- (sb1) node[midway, left] {\footnotesize $50$};
\draw (mb2) -- (sb1) node[midway, left] {\footnotesize $50$};
\draw (mb3) -- (sb1) node[midway, right] {\footnotesize $50$};
\draw (mb1) -- (sb2) node[midway, below] {\footnotesize $50$};
\draw (mb2) -- (sb2) node[midway, right] {\footnotesize $50$};
\draw (mb3) -- (sb2) node[midway, right] {\footnotesize $50$};

\draw (mg1) -- (sg1) node[midway, left] {\footnotesize $50$};
\draw (mg2) -- (sg1) node[midway, right] {\footnotesize $50$};

\draw[color=violet,dashed] ($(sb1)!0.5!(sb2)$) ellipse (1.5 and 0.5);
\draw[color=violet,dashed] ($(mb1)!0.5!(mb3)+(0,0.25)$) ellipse (3 and 1);
\draw[color=teal,dashed] ($(mg1)!0.5!(mg2)$) ellipse (1.5 and 0.5);
\draw[color=teal,dashed] (sg1) circle (0.5);
\draw[color=olive,dashed] (mr1) circle (0.5);

\filldraw (mb1) circle (2pt);
\filldraw (mb2) circle (2pt);
\filldraw (mb3) circle (2pt);
\filldraw (mg1) circle (2pt);
\filldraw (mg2) circle (2pt);
\filldraw (mr1) circle (2pt);
\filldraw (sb1) circle (2pt);
\filldraw (sb2) circle (2pt);
\filldraw (sg1) circle (2pt);
\node at ($(sb1)-(0.8,0)$) {\(V'_1\)};
\node at ($(mb3)+(1.2,0)$) {\(V_1\)};
\node at ($(mg1)-(0.8,0)$) {\(V_2\)};
\node at ($(sg1)+(0.8,0)$) {\(V'_2\)};
\node at ($(mr1)+(0.8,0)$) {\(V_3\)};

\end{tikzpicture}
    \caption{The reduction from Weighted Budgeted Matching to Maximum Matching for an instance with \(\rho_1=\rho_2=\rho_3=1\). Here we let \(M=50\).}
    \label{fig:reduction}
\end{figure}

\begin{lemma}\label{lem:min-max}
If maximum-weight matching can be solved in time \(T\lrp{n,m}\), then Weighted Budgeted Matching can be solved in time \(T\lrp{2n,m+n^2}\).
\end{lemma}

\begin{proof}
For each group \(g\) we add \(\lrm{V_g'}\leq\lrm{V_g}\) vertices for a total of at most \(n\) additional vertices. We then add a complete bipartite graph with \(\lrm{V_g'}\cdot\lrm{V_g}\) edges for each color, which totals at most \(n^2\) additional edges.
\end{proof}

\begin{restatable}[\(\dagger\)]{lemma}{lemFeas}
\label{lem:feasibility}
In \( \mathcal{M}_H \), at least \( \rho_g \) vertices from each group \( V_g \) are matched by edges from \( G \).
\end{restatable}

\begin{restatable}[\(\dagger\)]{lemma}{lemOpt}
\label{lem:optimality}
The matching $\mathcal{M} = \mathcal{M}_H \cap E(G)$ with original edge weights has total weight equal to the optimum of the Weighted Budgeted Matching problem.
\end{restatable}

To speed up the algorithm, we do a reduction to the problem of finding a max weight subgraph with upper bounds on degrees ($f$-factor). We contract the set $V'_g$ into a single node $V'_g$ with edges to all vertices in $V_g$ of weight $M$ and give a degree constraint of $|V_g|-\rho_g$ to $V'_g$. All nodes of $V$ will have a degree constraint of 1.
This new graph has only $n+\omega$ vertices, and at most $n$ extra edges, preserving the sparsity of the graph. We now use the algorithm by Gabow \cite{Gabow18} (see also \cite{GabowS21a}) which gives $O(nm+n^2 \log n)$ running time.

Lemmas \ref{lem:PEC-WBM}, \ref{lem:min-max}, \ref{lem:feasibility} and \ref{lem:optimality} together imply \Cref{thm:EC}.

\section{Conclusion}
While our algorithm achieves an $(f+\epsilon)$-approximation for the hard-capacitated problem, the algorithm for the weighted soft-capacitated (SC-PVC) variant has an approximation factor of $f+1$. Bridging this gap by designing an $f$-approximation algorithm for the weighted SC-PVC problem remains an intriguing open question.  A natural question is about the dependence on $\omega$ in the run-time. Is it possible to derive an $f$ approximation in time $O(f(\omega) n^c)$ so that the dependence on $\omega$ in the run-time is not $O(n^\omega)$? We already know that the problem is Set Cover hard for unbounded $\omega$. The same runtime issue arises in the $k$-center problem as mentioned earlier (in related work).

Another interesting direction for future research is to develop an efficient combinatorial algorithm for the Partition Vertex Cover problem, even in the uncapacitated case. While such primal-dual methods exist for the related Partial Vertex Cover problem \cite{Mestre09-VC}, developing one for the partition generalization is an open question.

For the distinct but related problem of maximizing edge coverage given a budget on the cardinality of the chosen vertex cover, Marx \cite{Marx08} shows how to obtain an FPT approximation scheme for unweighted simple graphs for $\omega=1$. It will be interesting to obtain similar results for arbitrary $\omega$.

\newpage

\bibliography{bibliography}

@article{partition-VC,
  author       = {Suman Kalyan Bera and
                  Shalmoli Gupta and
                  Amit Kumar and
                  Sambuddha Roy},
  title        = {Approximation algorithms for the partition vertex cover problem},
  journal      = {Theor. Comput. Sci.},
  volume       = {555},
  pages        = {2--8},
  year         = {2014},
  url          = {https://doi.org/10.1016/j.tcs.2014.04.006},
  doi          = {10.1016/J.TCS.2014.04.006},
  bibsource    = {dblp computer science bibliography, https://dblp.org}
}

@article{colorful-vertex-edge,
  author       = {Sayan Bandyapadhyay and
                  Aritra Banik and
                  Sujoy Bhore},
  title        = {On Colorful Vertex and Edge Cover Problems},
  journal      = {Algorithmica},
  volume       = {85},
  number       = {12},
  pages        = {3816--3827},
  year         = {2023},
  url          = {https://doi.org/10.1007/s00453-023-01164-6},
  doi          = {10.1007/S00453-023-01164-6},
  bibsource    = {dblp computer science bibliography, https://dblp.org}
}

@inproceedings{Karp-reducibility-72,
  author       = {Richard M. Karp},
  editor       = {Raymond E. Miller and
                  James W. Thatcher},
  title        = {Reducibility Among Combinatorial Problems},
  booktitle    = {Proceedings of a symposium on the Complexity of Computer Computations,
                  held March 20-22, 1972, at the {IBM} Thomas J. Watson Research Center,
                  Yorktown Heights, New York, {USA}},
  series       = {The {IBM} Research Symposia Series},
  pages        = {85--103},
  publisher    = {Plenum Press, New York},
  year         = {1972},
  url          = {https://doi.org/10.1007/978-1-4684-2001-2\_9},
  doi          = {10.1007/978-1-4684-2001-2\_9},
  bibsource    = {dblp computer science bibliography, https://dblp.org}
}

@article{hochbaum-VC-82,
  author       = {Dorit S. Hochbaum},
  title        = {Approximation Algorithms for the Set Covering and Vertex Cover Problems},
  journal      = {{SIAM} J. Comput.},
  volume       = {11},
  number       = {3},
  pages        = {555--556},
  year         = {1982},
  url          = {https://doi.org/10.1137/0211045},
  doi          = {10.1137/0211045},
  bibsource    = {dblp computer science bibliography, https://dblp.org}
}

@inproceedings{bar-yehuda-even-local-83,
  author       = {Reuven Bar{-}Yehuda and
                  Shimon Even},
  editor       = {Manfred Nagl and
                  J{\"{u}}rgen Perl},
  title        = {A Local-Ratio Theorem for Approximating the Weighted Vertex Cover
                  Problem},
  booktitle    = {Proceedings of the {WG} '83, International Workshop on Graphtheoretic
                  Concepts in Computer Science, June 16-18, 1983, Haus Ohrbeck, Germany},
  pages        = {17--28},
  publisher    = {Universit{\"{a}}tsverlag Rudolf Trauner, Linz},
  year         = {1983},
  bibsource    = {dblp computer science bibliography, https://dblp.org}
}

@article{bar-yehuda-even-81,
  author       = {Reuven Bar{-}Yehuda and
                  Shimon Even},
  title        = {A Linear-Time Approximation Algorithm for the Weighted Vertex Cover
                  Problem},
  journal      = {J. Algorithms},
  volume       = {2},
  number       = {2},
  pages        = {198--203},
  year         = {1981},
  url          = {https://doi.org/10.1016/0196-6774(81)90020-1},
  doi          = {10.1016/0196-6774(81)90020-1},
  bibsource    = {dblp computer science bibliography, https://dblp.org}
}

@article{nemhauser-trotter-75,
  author       = {George L. Nemhauser and
                  Leslie E. Trotter Jr.},
  title        = {Vertex packings: Structural properties and algorithms},
  journal      = {Math. Program.},
  volume       = {8},
  number       = {1},
  pages        = {232--248},
  year         = {1975},
  url          = {https://doi.org/10.1007/BF01580444},
  doi          = {10.1007/BF01580444},
  bibsource    = {dblp computer science bibliography, https://dblp.org}
}

@article{gavril-72,
author = {Gavril, F\u{a}nic\u{a}},
title = {Algorithms for Minimum Coloring, Maximum Clique, Minimum Covering by Cliques, and Maximum Independent Set of a Chordal Graph},
journal = {SIAM Journal on Computing},
volume = {1},
number = {2},
pages = {180-187},
year = {1972},
doi = {10.1137/0201013},

URL = { 
    
        https://doi.org/10.1137/0201013
    
    

},
eprint = { 
    
        https://doi.org/10.1137/0201013
    
    

}
}

@inproceedings{bshouty-burroughs-98,
  author       = {Nader H. Bshouty and
                  Lynn Burroughs},
  editor       = {Michel Morvan and
                  Christoph Meinel and
                  Daniel Krob},
  title        = {Massaging a Linear Programming Solution to Give a 2-Approximation
                  for a Generalization of the Vertex Cover Problem},
  booktitle    = {{STACS} 98, 15th Annual Symposium on Theoretical Aspects of Computer
                  Science, Paris, France, February 25-27, 1998, Proceedings},
  series       = {Lecture Notes in Computer Science},
  volume       = {1373},
  pages        = {298--308},
  publisher    = {Springer},
  year         = {1998},
  url          = {https://doi.org/10.1007/BFb0028569},
  doi          = {10.1007/BFB0028569},
  bibsource    = {dblp computer science bibliography, https://dblp.org}
}

@inproceedings{bar-yehuda-PVC-99,
  author       = {Reuven Bar{-}Yehuda},
  editor       = {Robert Endre Tarjan and
                  Tandy J. Warnow},
  title        = {Using Homogenous Weights for Approximating the Partial Cover Problem},
  booktitle    = {Proceedings of the Tenth Annual {ACM-SIAM} Symposium on Discrete Algorithms,
                  17-19 January 1999, Maryland, {USA}},
  pages        = {71--75},
  publisher    = {{ACM/SIAM}},
  year         = {1999},
  url          = {http://dl.acm.org/citation.cfm?id=314500.314533},
  bibsource    = {dblp computer science bibliography, https://dblp.org}
}

@inproceedings{gandhi-PVC-01,
  author       = {Rajiv Gandhi and
                  Samir Khuller and
                  Aravind Srinivasan},
  editor       = {Fernando Orejas and
                  Paul G. Spirakis and
                  Jan van Leeuwen},
  title        = {Approximation Algorithms for Partial Covering Problems},
  booktitle    = {Automata, Languages and Programming, 28th International Colloquium,
                  {ICALP} 2001, Crete, Greece, July 8-12, 2001, Proceedings},
  series       = {Lecture Notes in Computer Science},
  volume       = {2076},
  pages        = {225--236},
  publisher    = {Springer},
  year         = {2001},
  url          = {https://doi.org/10.1007/3-540-48224-5\_19},
  doi          = {10.1007/3-540-48224-5\_19},
  bibsource    = {dblp computer science bibliography, https://dblp.org}
}

@article{guha-capVC-03,
  author       = {Sudipto Guha and
                  Refael Hassin and
                  Samir Khuller and
                  Einat Or},
  title        = {Capacitated vertex covering},
  journal      = {J. Algorithms},
  volume       = {48},
  number       = {1},
  pages        = {257--270},
  year         = {2003},
  url          = {https://doi.org/10.1016/S0196-6774(03)00053-1},
  doi          = {10.1016/S0196-6774(03)00053-1},
  bibsource    = {dblp computer science bibliography, https://dblp.org}
}

@article{kortsarz-hardcapVC-06,
  author       = {Rajiv Gandhi and
                  Eran Halperin and
                  Samir Khuller and
                  Guy Kortsarz and
                  Aravind Srinivasan},
  title        = {An improved approximation algorithm for vertex cover with hard capacities},
  journal      = {J. Comput. Syst. Sci.},
  volume       = {72},
  number       = {1},
  pages        = {16--33},
  year         = {2006},
  url          = {https://doi.org/10.1016/j.jcss.2005.06.004},
  doi          = {10.1016/J.JCSS.2005.06.004},
  bibsource    = {dblp computer science bibliography, https://dblp.org}
}

@article{chuzhoy-hardcapVC-06,
  author       = {Julia Chuzhoy and
                  Joseph Naor},
  title        = {Covering Problems with Hard Capacities},
  journal      = {{SIAM} J. Comput.},
  volume       = {36},
  number       = {2},
  pages        = {498--515},
  year         = {2006},
  url          = {https://doi.org/10.1137/S0097539703422479},
  doi          = {10.1137/S0097539703422479},
  bibsource    = {dblp computer science bibliography, https://dblp.org}
}

@inproceedings{bandyapadhyay_constant_2019,
  author       = {Sayan Bandyapadhyay and
                  Tanmay Inamdar and
                  Shreyas Pai and
                  Kasturi R. Varadarajan},
  editor       = {Michael A. Bender and
                  Ola Svensson and
                  Grzegorz Herman},
  title        = {A Constant Approximation for Colorful k-Center},
  booktitle    = {27th Annual European Symposium on Algorithms, {ESA} 2019, September
                  9-11, 2019, Munich/Garching, Germany},
  series       = {LIPIcs},
  volume       = {144},
  pages        = {12:1--12:14},
  publisher    = {Schloss Dagstuhl - Leibniz-Zentrum f{\"{u}}r Informatik},
  year         = {2019},
  url          = {https://doi.org/10.4230/LIPIcs.ESA.2019.12},
  doi          = {10.4230/LIPICS.ESA.2019.12},
  bibsource    = {dblp computer science bibliography, https://dblp.org}
}

@article{anegg_technique_2022,
  author       = {Georg Anegg and
                  Haris Angelidakis and
                  Adam Kurpisz and
                  Rico Zenklusen},
  title        = {A technique for obtaining true approximations for k-center with covering
                  constraints},
  journal      = {Math. Program.},
  volume       = {192},
  number       = {1},
  pages        = {3--27},
  year         = {2022},
  url          = {https://doi.org/10.1007/s10107-021-01645-y},
  doi          = {10.1007/S10107-021-01645-Y},
  bibsource    = {dblp computer science bibliography, https://dblp.org}
}

@article{jia_fair_2022,
  author       = {Xinrui Jia and
                  Kshiteej Sheth and
                  Ola Svensson},
  title        = {Fair colorful k-center clustering},
  journal      = {Math. Program.},
  volume       = {192},
  number       = {1},
  pages        = {339--360},
  year         = {2022},
  url          = {https://doi.org/10.1007/s10107-021-01674-7},
  doi          = {10.1007/S10107-021-01674-7},
  bibsource    = {dblp computer science bibliography, https://dblp.org}
}

@article{Mestre09-VC,
  author       = {Juli{\'{a}}n Mestre},
  title        = {A Primal-Dual Approximation Algorithm for Partial Vertex Cover: Making
                  Educated Guesses},
  journal      = {Algorithmica},
  volume       = {55},
  number       = {1},
  pages        = {227--239},
  year         = {2009},
  url          = {https://doi.org/10.1007/s00453-007-9003-z},
  doi          = {10.1007/S00453-007-9003-Z},
  bibsource    = {dblp computer science bibliography, https://dblp.org}
}

@inproceedings{gabow-90,
  author       = {Harold N. Gabow},
  editor       = {David S. Johnson},
  title        = {Data Structures for Weighted Matching and Nearest Common Ancestors
                  with Linking},
  booktitle    = {Proceedings of the First Annual {ACM-SIAM} Symposium on Discrete Algorithms,
                  22-24 January 1990, San Francisco, California, {USA}},
  pages        = {434--443},
  publisher    = {{SIAM}},
  year         = {1990},
  url          = {http://dl.acm.org/citation.cfm?id=320176.320229},
  bibsource    = {dblp computer science bibliography, https://dblp.org}
}

@inproceedings{dabas2025flofair,
  author       = {Rajni Dabas and
                  Samir Khuller and
                  Emilie Rivkin},
  editor       = {Huijia (Rachel) Lin},
  title        = {Serving Clients Fairly: On Facility Location and k-Median with Fair
                  Outliers},
  booktitle    = {7th Symposium on Foundations of Responsible Computing, {FORC} 2026,
                  Harvard University, Cambridge, MA, USA, June 3-5, 2026},
  series       = {LIPIcs},
  volume       = {368},
  pages        = {9:1--9:19},
  publisher    = {Schloss Dagstuhl - Leibniz-Zentrum f{\"{u}}r Informatik},
  year         = {2026},
  url          = {https://doi.org/10.4230/LIPIcs.FORC.2026.9},
  doi          = {10.4230/LIPICS.FORC.2026.9},
  bibsource    = {dblp computer science bibliography, https://dblp.org}
}

@inproceedings{LiuLi-COCOON2025,
  author    = {Xiaofei Liu and Weidong Li},
  title     = {Approximation Algorithm for Prize-Collecting Hypergraph Vertex Cover with Fairness Constraints},
  booktitle = {Computing and Combinatorics (COCOON 2025)},
  series    = {Lecture Notes in Computer Science},
  volume    = {15984},
  pages     = {41--52},
  publisher = {Springer},
  address   = {Singapore},
  year      = {2026},
  doi       = {10.1007/978-981-95-0218-9_4}
}

@article{GabowS21a,
  author       = {Harold N. Gabow and
                  Piotr Sankowski},
  title        = {Algorithms for Weighted Matching Generalizations {II:} \emph{f}-factors
                  and the Special Case of Shortest Paths},
  journal      = {{SIAM} J. Comput.},
  volume       = {50},
  number       = {2},
  pages        = {555--601},
  year         = {2021},
  url          = {https://doi.org/10.1137/16M1106225},
  doi          = {10.1137/16M1106225},
  bibsource    = {dblp computer science bibliography, https://dblp.org}
}

@article{Gabow18,
  author       = {Harold N. Gabow},
  title        = {Data Structures for Weighted Matching and Extensions to \emph{b}-matching
                  and \emph{f}-factors},
  journal      = {{ACM} Trans. Algorithms},
  volume       = {14},
  number       = {3},
  pages        = {39:1--39:80},
  year         = {2018},
  url          = {https://doi.org/10.1145/3183369},
  doi          = {10.1145/3183369},
  bibsource    = {dblp computer science bibliography, https://dblp.org}
}

@inproceedings{SahaK12-vc,
  author       = {Barna Saha and
                  Samir Khuller},
  editor       = {Artur Czumaj and
                  Kurt Mehlhorn and
                  Andrew M. Pitts and
                  Roger Wattenhofer},
  title        = {Set Cover Revisited: Hypergraph Cover with Hard Capacities},
  booktitle    = {Automata, Languages, and Programming - 39th International Colloquium,
                  {ICALP} 2012, Warwick, UK, July 9-13, 2012, Proceedings, Part {I}},
  series       = {Lecture Notes in Computer Science},
  volume       = {7391},
  pages        = {762--773},
  publisher    = {Springer},
  year         = {2012},
  url          = {https://doi.org/10.1007/978-3-642-31594-7\_64},
  doi          = {10.1007/978-3-642-31594-7\_64},
  bibsource    = {dblp computer science bibliography, https://dblp.org}
}

@inproceedings{CheungGW14-vc,
  author       = {Wang Chi Cheung and
                  Michel X. Goemans and
                  Sam Chiu{-}wai Wong},
  editor       = {Chandra Chekuri},
  title        = {Improved Algorithms for Vertex Cover with Hard Capacities on Multigraphs
                  and Hypergraphs},
  booktitle    = {Proceedings of the Twenty-Fifth Annual {ACM-SIAM} Symposium on Discrete
                  Algorithms, {SODA} 2014, Portland, Oregon, USA, January 5-7, 2014},
  pages        = {1714--1726},
  publisher    = {{SIAM}},
  year         = {2014},
  url          = {https://doi.org/10.1137/1.9781611973402.124},
  doi          = {10.1137/1.9781611973402.124},
  bibsource    = {dblp computer science bibliography, https://dblp.org}
}

@inproceedings{Kao17-vc,
  author       = {Mong{-}Jen Kao},
  editor       = {Philip N. Klein},
  title        = {Iterative Partial Rounding for Vertex Cover with Hard Capacities},
  booktitle    = {Proceedings of the Twenty-Eighth Annual {ACM-SIAM} Symposium on Discrete
                  Algorithms, {SODA} 2017, Barcelona, Spain, January
                  16-19},
  pages        = {2638--2653},
  publisher    = {{SIAM}},
  year         = {2017},
  url          = {https://doi.org/10.1137/1.9781611974782.174},
  doi          = {10.1137/1.9781611974782.174},
  bibsource    = {dblp computer science bibliography, https://dblp.org}
}

@article{KaoSLL19-pvc,
  author       = {Mong{-}Jen Kao and
                  Jia{-}Yau Shiau and
                  Ching{-}Chi Lin and
                  D. T. Lee},
  title        = {Tight approximation for partial vertex cover with hard capacities},
  journal      = {Theor. Comput. Sci.},
  volume       = {778},
  pages        = {61--72},
  year         = {2019},
  url          = {https://doi.org/10.1016/j.tcs.2019.01.027},
  doi          = {10.1016/J.TCS.2019.01.027},
  bibsource    = {dblp computer science bibliography, https://dblp.org}
}

@inproceedings{Wong17-vc,
  author       = {Sam Chiu{-}wai Wong},
  editor       = {Philip N. Klein},
  title        = {Tight Algorithms for Vertex Cover with Hard Capacities on Multigraphs
                  and Hypergraphs},
  booktitle    = {Proceedings of the Twenty-Eighth Annual {ACM-SIAM} Symposium on Discrete
                  Algorithms, {SODA} 2017, Barcelona, Spain,
                  16-19},
  pages        = {2626--2637},
  publisher    = {{SIAM}},
  year         = {2017},
  url          = {https://doi.org/10.1137/1.9781611974782.173},
  doi          = {10.1137/1.9781611974782.173},
  bibsource    = {dblp computer science bibliography, https://dblp.org}
}

@article{DabasKR-vcec,
  author       = {Rajni Dabas and
                  Samir Khuller and
                  Emilie Rivkin},
  title        = {Weighted Partition Vertex and Edge Cover},
  journal      = {CoRR},
  volume       = {abs/2508.13055},
  year         = {2025},
  url          = {https://doi.org/10.48550/arXiv.2508.13055},
  doi          = {10.48550/ARXIV.2508.13055},
  eprinttype    = {arXiv},
  eprint       = {2508.13055},
  bibsource    = {dblp computer science bibliography, https://dblp.org}
}

@article{Bar-YehudaFMR10,
  author       = {Reuven Bar{-}Yehuda and
                  Guy Flysher and
                  Juli{\'{a}}n Mestre and
                  Dror Rawitz},
  title        = {Approximation of Partial Capacitated Vertex Cover},
  journal      = {{SIAM} J. Discret. Math.},
  volume       = {24},
  number       = {4},
  pages        = {1441--1469},
  year         = {2010},
  url          = {https://doi.org/10.1137/080728044},
  doi          = {10.1137/080728044},
  bibsource    = {dblp computer science bibliography, https://dblp.org}
}

@inproceedings{BajpaiCK25-fl,
  author       = {Tanvi Bajpai and
                  Chandra Chekuri and
                  Pooja Kulkarni},
  editor       = {Alina Ene and
                  Eshan Chattopadhyay},
  title        = {Covering a Few Submodular Constraints and Applications},
  booktitle    = {Approximation, Randomization, and Combinatorial Optimization. Algorithms
                  and Techniques, {APPROX/RANDOM} 2025, August 11-13, 2025, Berkeley,
                  CA, {USA}},
  series       = {LIPIcs},
  volume       = {353},
  pages        = {25:1--25:22},
  publisher    = {Schloss Dagstuhl - Leibniz-Zentrum f{\"{u}}r Informatik},
  year         = {2025},
  url          = {https://doi.org/10.4230/LIPIcs.APPROX/RANDOM.2025.25},
  doi          = {10.4230/LIPICS.APPROX/RANDOM.2025.25},
  bibsource    = {dblp computer science bibliography, https://dblp.org}
}

@misc{Ingold2016Amazon,
  author       = {David Ingold and Spencer Soper},
  title        = {Amazon Doesn’t Consider the Race of Its Customers. Should It?},
  year         = {2016},
  month        = {4},
  day          = {21},
  publisher    = {Bloomberg},
  url          = {https://www.bloomberg.com/graphics/2016-amazon-same-day/},
  note         = {Accessed: 2026-01-20}
}

@article{Dill2021Equity,
  author       = {Jennifer Dill and Nathan McNeil},
  title        = {Are Shared Vehicles Shared by All? A Review of Equity and Vehicle Sharing},
  journal      =  {Journal of Planning Literature},
  year         = {2020},
volume = {36 (1)},
pages        = {5-30},
 doi          = {https://doi.org/10.1177/0885412220966732}
}

@article{Marx08,
  author       = {D{\'{a}}niel Marx},
  title        = {Parameterized Complexity and Approximation Algorithms},
  journal      = {Comput. J.},
  volume       = {51},
  number       = {1},
  pages        = {60--78},
  year         = {2008},
  url          = {https://doi.org/10.1093/comjnl/bxm048},
  doi          = {10.1093/COMJNL/BXM048},
  bibsource    = {dblp computer science bibliography, https://dblp.org}
}

@article{HongKao2022,
  author    = {Eunpyeong Hong and Mong-Jen Kao},
  title     = {Approximation Algorithm for Vertex Cover with Multiple Covering Constraints},
  journal   = {Algorithmica},
  volume    = {84},
  number    = {1},
  pages     = {1--12},
  year      = {2022},
  publisher = {Springer},
  doi       = {10.1007/s00453-021-00885-w}
}

@article{KhotRegev2008,
author = {Khot, Subhash and Regev, Oded},
title = {Vertex cover might be hard to approximate to within 2-$\epsilon$},
year = {2008},
issue_date = {May, 2008},
publisher = {Academic Press, Inc.},
address = {USA},
volume = {74},
number = {3},
issn = {0022-0000},
url = {https://doi.org/10.1016/j.jcss.2007.06.019},
doi = {10.1016/j.jcss.2007.06.019},
journal = {J. Comput. Syst. Sci.},
month = may,
pages = {335–349},
numpages = {15}
}

@inproceedings{DinurS14SETCOVER,
  author       = {Irit Dinur and
                  David Steurer},
  editor       = {David B. Shmoys},
  title        = {Analytical approach to parallel repetition},
  booktitle    = {Symposium on Theory of Computing, {STOC} 2014, New York, NY, USA,
                  May 31 - June 03, 2014},
  pages        = {624--633},
  publisher    = {{ACM}},
  year         = {2014},
  url          = {https://doi.org/10.1145/2591796.2591884},
  doi          = {10.1145/2591796.2591884},
  bibsource    = {dblp computer science bibliography, https://dblp.org}
}

\appendix

\section{Remaining Proofs from Section \ref{sec:sc-pvc}}

\lemFlow*

\begin{proof}
We present the construction for completeness. Consider the input instance and a solution \(\lra{x,y}\) with \(x\) integral. We construct a network flow instance \(F\). The nodes of \(F\) are \(s\cup E\cup\lrp{V\cup\lrs{U_1,\ldots,U_\omega}}\cup t\). The source and sink are \(s\) and \(t\), respectively. The nodes \(U_1,\ldots,U_\omega\) correspond to the uncovered sets for each group. We define the arcs and capacities:
\begin{itemize}
\item \(\lrp{s,e}\) for each \(e\in E\) with capacity 1;
\item \(\lrp{e,v}\) for each \(v\in e\) with capacity 1;
\item \(\lrp{e,U_g}\) where \(e\in E_g\) with capacity 1;
\item \(\lrp{v,t}\) for each \(v\in V\) with capacity \(k_vx_v\);
\item \(\lrp{U_g,t}\) for each \(g=1,\ldots,\omega\) with capacity \(|E_g|-\rho_g\).
\end{itemize}

We know that \(\lra{x,y}\) is feasible, so there must be a flow of \(\lrm{E}\) in \(F\) from \(s\) to \(t\). The capacity on each arc is integral, so there exists and integral flow of value \(\lrm{E}\) from \(s\) to \(t\). Define \(\hy_{e,v}=1\) (meaning that \(\phi(e)=v\)) if the flow \(\lrp{e,v}\) is 1 and \(\hy_{e,v}=0\) otherwise. If an edge goes unassigned, meaning that the flow from \(e\) to \(U_g\) is \(1\), \(\phi(e)=\perp\). Then, \(\lra{x,\hy}\) is a feasible integral solution.
\end{proof}

\lemH*

\begin{proof}
We only consider vertices affected by rounding cases (1) and (3b). For (1), the value of the edge variable is unchanged. In the case of (3b), we know that for all \(v\in e\cap L\), \(y_{e,v}<\frac{c_e}{f+1}\). We consider the worst possible case for rounding here. An endpoint in \(H\) handles \(1+\frac{c^\ell_e}{c^h_e}\) times more edge coverage, which increases as the mass \(c_e^\ell\) increases. This added responsibility is therefore maximized when the edge contains \(f-1\) endpoints in \(L\) and a single endpoint in \(v_h\in H\). Hence it suffices to analyze this worst case. Then,
\[
\hy_{e,v_h} = c_e = y_{e,v_h}+\sum_{v\in e\cap L}y_{e,v} < y_{e,v_h}+(f-1)\frac{c_e}{f+1}.
\]
We can then conclude that 
\[
\hy_{e,v_h} = c_e < \frac{f+1}{2}y_{e,v_h}.
\]
Taking the sum over edges incident to \(v\) gives us the first inequality.

Now, we bound \(\hx_v\). For \(v\in H\cap R\), we consider two cases. If \(\frac{1}{f+1}\leq x_v\leq\frac{2}{f+1}\), then 
\[
\hx_v = \lrceil{\frac{f+1}{2}x_v} = 1 \leq (f+1)x_v.
\]
Otherwise, we have \(x_v>\frac{2}{f+1}\). In that case, 
\[
\hx_v = \lrceil{\frac{f+1}{2}x_v} < \frac{f+1}{2}x_v+1 < (f+1)x_v.
\]

For \(v\in M\), \(x_v\in\N\). Therefore,
\[
\hx_v = \lrceil{\frac{f+1}{2}x_v} \leq \lrceil {f \cdot x_v} = f\cdot x_v.
\] where the first inequality follows simply because $\frac{f+1}{2} \leq f$ for $f\geq 2$ and last equality follows because both $f$ and $x_v$ are integers.
\end{proof}

\section{Remaining Proofs from Section \ref{sec:w-pec}}

\lemFeas*

\begin{proof}
Let \( \mathcal{M}^* \subseteq G \) be an optimal solution to the budgeted matching problem. We now construct a matching \( \mathcal{M}_H^* \) in \( H \) by including all edges from \( \mathcal{M}^* \), and  for each vertex in \( V_g \) not matched in \( \mathcal{M}^* \), adding an auxiliary edge to a unique free node in \( V'_g \).

Let \( k_g^* \) be the number of vertices in \( V_g \) matched by \( \mathcal{M}^* \). Since \( k_g^* \ge \rho_g \), we have, \(|V_g| - k_g^* \le |V_g| - \rho_g = |V'_g|\), so there are enough auxiliary vertices in \( V'_g \) to complete the matching without conflict. Thus, \( \mathcal{M}_H^* \) is a feasible matching in \( H \) that matches all the vertices in $V_g$ for every $g$. This also ensures that the cost of the maximum matching in $H$ is strictly greater than $(n-1)M$ since each matched node in $V$ contributes $M$ to the weight of the matching, and even after subtracting the weight of a few edges we have a solution of weight between $nM$ and $(n-1)M$. Any matching that does not match all nodes of $G$ (that also belong to $H$) cannot have weight larger than $(n-1)M$.

Now consider the maximum weight matching \( \mathcal{M}_H \) in \( H \). Suppose, for contradiction, that \( \mathcal{M}_H \) matches fewer than \( \rho_g \) vertices from some group \( V_g \) via edges from \( G \). The number of vertices matched by auxiliary edges is at most \( |V'_g| = |V_g| - \rho_g \). Therefore, at least one vertex \( u \in V_g \) must be unmatched by \( \mathcal{M}_H \). Therefore, for large $M$, \( \mathcal{M}^*_H \) is better than \( \mathcal{M}_H \), which is a contradiction.
\end{proof}

\lemOpt*

\begin{proof}
We assume $n=\vert{}V\vert{}$ is even. This is perfectly fine because the reduction from Partition Edge Cover in Lemma 15 yields a Budgeted Matching instance with exactly $2n$ vertices, ensuring it is always an even-size instance. Let $\mathcal{M}^* \subseteq G$ be an optimal solution to the budgeted matching problem with total weight $\text{OPT}$. As in proof of Lemma~\ref{lem:feasibility}, we can extend $\mathcal{M}^*$ to a matching $\mathcal{M}_H^*$ in $H$ by adding auxiliary edges of weight $M$ to unmatched vertices in $V$. Since all $n = |V|$ vertices in $V$ are matched in $\mathcal{M}_H^*$, and exactly $n - 2|\mathcal{M}^*|$ auxiliary edges are used, the total weight of $\mathcal{M}_H^*$ is:
\[
\sum_{e \in \mathcal{M}_H^*} \weh_e = \sum_{e \in \mathcal{M}^*}(2M - \we_e) + (n - 2|\mathcal{M}^*|)M = nM - \text{OPT}.
\]

Let $\mathcal{M}_H$ be a maximum weight matching in $H$, and let $\mathcal{M} = \mathcal{M}_H \cap E(G)$. 
Note that $\mathcal{M}_H$ also matches all $n$ vertices in $V$. Therefore,
\[
\sum_{e \in \mathcal{M}_H} \weh_e= \sum_{e \in \mathcal{M}}(2M - \we_e) + (n - 2|\mathcal{M}|)M = nM - \sum_{e \in \mathcal{M}} \we_e.
\]

Since $\mathcal{M}_H$ maximizes the weight,
\[
\sum_{e \in \mathcal{M}_H} \weh_e \ge \sum_{e \in \mathcal{M}_H^*} \weh_e  \Rightarrow \sum_{e \in \mathcal{M}} \we_e \le \text{OPT}.
\]

By Lemma~\ref{lem:feasibility}, $\mathcal{M}$ is a feasible solution to the budgeted matching problem, so equality must hold and $\mathcal{M}$ is optimal.
\end{proof}

\end{document}